\documentclass{article}

\usepackage{epsfig}
\usepackage{amsfonts}
\usepackage{amssymb}
\usepackage{amsmath}
\usepackage{amsthm}
\usepackage{latexsym}
\usepackage{subfigure}
\usepackage{verbatim}
\usepackage{bbm}
\usepackage{comment}
\usepackage{url}
\usepackage{multicol}
\usepackage{array}
\usepackage{color}
\usepackage{natbib}
\catcode`\@=11%

\usepackage{natbib}

\def\independenT#1#2{\mathrel{\setbox0\hbox{$#1#2$}%
\copy0\kern-\wd0\mkern4mu\box0}}

\newtheorem{theorem}{Theorem}
\newtheorem{proposition}[theorem]{Proposition}
\newtheorem{lemma}[theorem]{Lemma}

\newcommand{\app}[1]{Appendix~\ref{app:#1}}

\newcommand{\eq}[1]{Eq.~(\ref{eq:#1})}
\newcommand{\eqs}[1]{Eqs.~(\ref{eq:#1})}
\newcommand{\eqss}[1]{(\ref{eq:#1})}
\newcommand{\fig}[1]{Figure~\ref{fig:#1}}
\newcommand{\figs}[1]{Figures~\ref{fig:#1}}
\newcommand{\figss}[1]{\ref{fig:#1}}
\newcommand{\lem}[1]{Lemma~\ref{lem:#1}}
\newcommand{\lems}[1]{Lemmas~\ref{lem:#1}}
\newcommand{\lemss}[1]{\ref{lem:#1}}
\newcommand{\prop}[1]{Proposition~\ref{prop:#1}}
\newcommand{\mysec}[1]{Section~\ref{sec:#1}}
\newcommand{\thm}[1]{Theorem~\ref{thm:#1}}

\newcommand{\bep}{\textrm{BeP}}
\newcommand{\bern}{\textrm{Bern}}
\newcommand{\tb}{\textrm{Beta}}
\newcommand{\bp}{\textrm{BP}}
\newcommand{\bpbep}{\textrm{BP-BeP}}
\newcommand{\tdp}{\textrm{DP}}
\newcommand{\ga}{\textrm{Ga}}

\newcommand{\pois}{\textrm{Pois}}
\newcommand{\py}{\textrm{PY}}

\newcommand{\mbe}{\mathbb{E}}
\newcommand{\mbp}{\mathbb{P}}

\allowdisplaybreaks
\title{Beta processes, stick-breaking, and power laws}
\author{Tamara Broderick\footnote{Corresponding author: \texttt{tab@stat.berkeley.edu}}, Michael I. Jordan, Jim Pitman}
\date{} 

\begin{document}
\maketitle

\begin{abstract}
The beta-Bernoulli process provides a Bayesian nonparametric prior
for models involving collections of binary-valued features.  
A draw from the beta process yields an infinite collection of
probabilities in the unit interval, and a draw from the Bernoulli
process turns these into binary-valued features.  Recent work has
provided stick-breaking representations for the beta
process analogous to the well-known stick-breaking
representation for the Dirichlet process.  We derive
one such stick-breaking representation directly from the characterization
of the beta process as a completely random measure.  
This approach motivates a three-parameter generalization of the
beta process, and we study the power laws that can be obtained from
this generalized beta process.  We present a posterior inference
algorithm for the beta-Bernoulli process that exploits the stick-breaking
representation, and we present experimental results for a discrete
factor-analysis model.
\end{abstract}

\section{Introduction}

Large data sets are often heterogeneous, arising as amalgams from 
underlying sub-populations.  The analysis of large data sets thus
often involves some form of stratification in which groupings are identified 
that are more homogeneous than the original data.  While this can sometimes be 
done on the basis of explicit covariates, it is also commonly the case that 
the groupings are captured via discrete latent variables that are to be 
inferred as part of the analysis.  Within a Bayesian framework, there 
are two widely employed modeling motifs for problems of this kind.  The 
first is the \emph{Dirichlet-multinomial motif}, which is based on the
assumption that there are $K$ ``clusters'' that are assumed to be mutually 
exclusive and exhaustive, such that allocations of data to clusters can be 
modeled via a multinomial random variable whose parameter vector is 
drawn from a Dirichlet distribution.  A second motif is the 
\emph{beta-Bernoulli motif}, where a collection of $K$ binary 
``features'' are used to describe the data, and where each feature 
is modeled as a Bernoulli random variable whose parameter is obtained 
from a beta distibution.  The latter motif can be converted to the
former in principle---we can view particular patterns of ones and 
zeros as defining a cluster, thus obtaining $M = 2^K$ clusters in total.
But in practice models based on the Dirichlet-multinomial motif 
typically require $O(M)$ additional parameters in the likelihood, 
whereas those based on the beta-Bernoulli motif typically require 
only $O(K)$ additional parameters.  Thus, if the combinatorial structure 
encoded by the binary features captures real structure in the data,
then the beta-Bernoulli motif can make more efficient usage of its
parameters.

The Dirichlet-multinomial motif can be extended to a stochastic 
process known as the \emph{Dirichlet process}.  A draw from a Dirichlet 
process is a random probability measure that can be represented as 
follows~\citep{mccloskey:1965:model,patil:1977:diversity,
ferguson:1973:bayesian,sethuraman:1994:constructive}:
\begin{equation}
G = \sum_{i=1}^{\infty} \pi_{i} \delta_{\psi_{i}},
\label{eq:DP}
\end{equation}
where $\delta_{\psi_{i}}$ represents an atomic measure at location
$\psi_i$, where both the $\{\pi_{i}\}$ and the $\{\psi_{i}\}$ are random,
and where the $\{\pi_{i}\}$ are nonnegative and sum to one (with probability 
one).  Conditioning on $G$ and drawing $N$ values independently from 
$G$ yields a collection of $K$ distinct values, where $K \leq N$ is 
random and grows (in expectation) at rate $O(\log N)$.  Treating 
these distinct values as indices of clusters, we obtain a model in 
which the number of clusters is random and subject to posterior inference.

A great deal is known about the Dirichlet process---there are 
direct connections between properties of $G$ as a random measure 
(e.g., it can be obtained from a Poisson point process), properties 
of the sequence of values $\{\pi_i\}$ (they can be obtained from a 
``stick-breaking process''), and properties of the collection of distinct values 
obtained by sampling from $G$ (they are characterized by a stochastic 
process known as the \emph{Chinese restaurant process}).  These connections 
have helped to place the Dirichlet process at the center of Bayesian 
nonparametrics, driving the development of a wide variety of inference 
algorithms for models based on Dirichlet process priors and suggesting 
a range of generalizations~\citep[e.g.][]{maceachern:1999:dependent,
ishwaran:2001:gibbs,walker:2007:sampling,kalli:2010:slice}.  

It is also possible to extend the beta-Bernoulli motif to a Bayesian 
nonparametric framework, and there is a growing literature on this topic.
The underlying stochastic process is the \emph{beta process}, which is 
an instance of a family of random measures known as \emph{completely
random measures}~\citep{kingman:1967:completely}.  The beta process was first studied 
in the context of survival analysis by~\citet{hjort:1990:nonparametric},
where the focus is on modeling hazard functions via the random cumulative 
distribution function obtained by integrating the beta process.  
\citet{thibaux:2007:hierarchical} focused instead on the beta process 
realization itself, which can be represented as
\[
G = \sum_{i=1}^{\infty} q_{i} \delta_{\psi_{i}},
\]
where both the $q_{i}$ and the $\psi_{i}$ are random and where the $q_{i}$ 
are contained in the interval $(0,1)$.  This random measure can be viewed as
furnishing an infinite collection of coins, which, when tossed repeatedly,
yield a binary featural description of a set of entities in which the 
number of features with non-zero values is random.  Thus, the resulting 
\emph{beta-Bernoulli process} can be viewed as an infinite-dimensional version 
of the beta-Bernoulli motif.  Indeed, \citet{thibaux:2007:hierarchical} 
showed that by integrating out the random $q_{i}$ and $\psi_{i}$ one 
obtains---by analogy to the derivation of the Chinese restaurant process 
from the Dirichlet process---a combinatorial stochastic process known as 
the \emph{Indian buffet process}, previously studied by~\citet{griffiths:2006:infinite}, 
who derived it via a limiting process involving random binary matrices
obtained by sampling finite collections of beta-Bernoulli variables.

Stick-breaking representations of the Dirichlet process have been 
particularly important both for algorithmic development and for 
exploring generalizations of the Dirichlet process.  These
representations yield explicit recursive formulas for obtaining 
the weights $\{\pi_i\}$ in \eq{DP}.  In the case of the beta
process, explicit non-recursive representations can be obtained 
for the weights $\{q_i\}$, based on size-biased sampling~\citep{thibaux:2007:hierarchical}
and inverse L\'evy measure~\citep{wolpert:2004:reflecting,teh:2007:stick}.  
Recent work has also yielded recursive constructions that are 
more closely related to the stick-breaking representation of the 
Dirichlet process~\citep{teh:2007:stick,paisley:2010:stick}.  

Stick-breaking representations of the Dirichlet process permit
ready generalizations to stochastic processes that yield power-law
behavior (which the Dirichlet process does not), notably the
Pitman-Yor process~\citep{ishwaran:2001:gibbs,pitman:2006:combinatorial}.  
Power-law generalizations of the beta process have also been 
studied~\citep{teh:2009:indian} and stick-breaking-like representations 
derived.  These latter representations are, however, based on the 
non-recursive sized-biased sampling and inverse-L\'evy methods rather 
than the recursive representations of~\citet{teh:2007:stick}
and~\citet{paisley:2010:stick}.

\citet{teh:2007:stick} and~\citet{paisley:2010:stick} derived their 
stick-breaking representations of the beta process as limiting processes, 
making use of the derivation of the Indian buffet process 
by~\citet{griffiths:2006:infinite} as a limit of finite-dimensional
random matrices.  
In the current paper we show how to derive stick-breaking 
for the beta process directly from the underlying random 
measure.  This approach not only has the advantage of conceptual 
clarity (our derivation is elementary), but it also permits a 
unified perspective on various generalizations of the beta process 
that yield power-law behavior.\footnote{A similar measure-theoretic
derivation has been presented recently by~\citet{paisley:2011:stick}, 
who focus on applications to truncations of the beta process.}
We show in particular that it yields a power-law generalization
of the stick-breaking representation of~\citet{paisley:2010:stick}.

To illustrate our results in the context of a concrete application,
we study a discrete factor analysis model previously considered 
by~\citet{griffiths:2006:infinite} and \citet{paisley:2010:stick}.  
The model is of the form
\begin{equation} 
	\label{eq:factor}
	X = Z \Phi + E,
\end{equation}
where $X \in \mathbb{R}^{N \times P}$ is the data and $E \in \mathbb{R}^{N \times P}$
is an error matrix.  The matrix $\Phi \in \mathbb{R}^{K \times P}$ is a matrix
of factors, and $Z \in \mathbb{R}^{N \times K}$ is a binary matrix of factor
loadings.  The dimension $K$ is infinite, and thus the rows of $\Phi$ comprise 
an infinite collection of factors.  The matrix $Z$ is obtained 
via a draw from a beta-Bernoulli process; its $n$th row is an infinite binary vector 
of features (i.e., factor loadings) encoding which of the infinite collection
of factors are used in modeling the $n$th data point.

The remainder of the paper is organized as follows. We introduce the beta process, 
and its conjugate measure the Bernoulli process, in \mysec{model}. In order to 
consider stick-breaking and power law behavior in the beta-Bernoulli framework, 
we first review stick-breaking for the Dirichlet process in \mysec{dp_stick} 
and power laws in clustering models in \mysec{clustering_power}. We consider 
potential power laws that might exist in featural models in \mysec{power_mult}. 
Our main theoretical results come in the following two sections. First, 
in \mysec{bp_stick}, we provide a proof that the stick-breaking 
representation of~\citet{paisley:2010:stick}, expanded to include a third 
parameter, holds for a three-parameter extension of the beta process.
Our proof takes a measure-theoretic approach based on 
a Poisson process.  We then make use of the Poisson process framework to 
establish asymptotic power laws, with exact constants, for the three-parameter 
beta process in \mysec{types_i_ii_proofs}.
We also show, in \mysec{exp_proof}, that there are aspects of 
the beta-Bernoulli framework that cannot exhibit a power law.  We illustrate 
the asymptotic power laws on a simulated data set in \mysec{simulation}.
We present experimental results in 
\mysec{experiments}, and we present an MCMC algorithm for posterior inference
in \app{inference}. 

\section{The beta process and the Bernoulli process} \label{sec:model}

The beta process and the Bernoulli process are instances of the general 
family of random measures known as \emph{completely random 
measures}~\citep{kingman:1967:completely}.  A completely random 
measure $H$ on a probability space $(\Psi, {\cal S})$ is a random 
measure such that, for any disjoint measurable sets $A_{1},\ldots,A_{n} \in {\cal S}$, 
the random variables $H(A_{1}),\ldots,H(A_{n})$ are independent.

Completely random measures
can be obtained from an underlying 
Poisson point process.
Let $\nu(d\psi, du)$ denote a $\sigma$-finite measure\footnote{
The measure $\nu$ need not necessarily be $\sigma$-finite
to generate a completely random measure though we 
consider only $\sigma$-finite measures in this work.} 
on the product space $\Psi \times \mathbb{R}$.  Draw a realization from 
a Poisson point process with rate measure $\nu(d\psi, du)$.
This yields a set of points $\Pi = \{(\psi_{i},U_{i})\}_{i}$, 
where the index $i$ may range over a countable infinity.  Finally, 
construct a random measure as follows:
\begin{equation}
	\label{eq:crm_draw}
	B = \sum_{i=1}^{\infty} U_{i} \delta_{\psi_{i}},
\end{equation}
where $\delta_{\psi_{i}}$ denotes an atom at $\psi_i$.
This discrete random measure is such that for any measurable
set $T \in {\cal S}$,
$$
	B(T) = \sum_{i: \psi_{i} \in T} U_{i}.
$$
That $B$ is completely random follows from the Poisson point process 
construction.

In addition to the representation
obtained from a Poisson process, completely random measures may
include a deterministic measure and a set of atoms at fixed
locations. The component of the completely random measure generated from a
Poisson point process as described above is called the {\em ordinary
component}. As shown by~\citet{kingman:1967:completely}, 
completely random measures are essentially characterized by 
this representation. An example is shown in \fig{beta}.

\begin{figure}
\begin{center}
\includegraphics[width=0.45\textwidth]{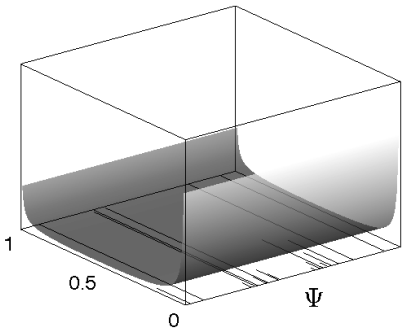}
\end{center}
\caption{\label{fig:beta}
The gray surface illustrates the rate density
in \eq{beta_density} corresponding to the beta process.  
The base measure $B_{0}$ is taken to be uniform on $\Psi$.  
The non-zero endpoints of the line segments plotted below the 
surface are a particular realization of the Poisson process, 
and the line segments themselves represent a realization of
the beta process.}
\end{figure}

The \emph{beta process}, denoted $B \sim \bp(\theta, B_{0})$, 
is an example of a completely random measure.  As long as the 
\emph{base measure} $B_{0}$ is continuous, which is our assumption 
here, $B$ has only an ordinary component with rate measure
\begin{equation}
	\label{eq:beta_density}
	\nu_{\bp}(d\psi, du) = \theta(\psi) u^{-1} (1-u)^{\theta(\psi)-1} \; du \; B_{0}(d \psi), \quad \psi \in \Psi, u \in [0,1],
\end{equation}
where $\theta$ is a positive function on $\Psi$.
The function $\theta$
is called the {\em concentration function}~\citep{hjort:1990:nonparametric}.
In the remainder we follow~\citet{thibaux:2007:hierarchical} in taking $\theta$ 
to be a real-valued constant and refer to it as the {\em concentration parameter}.  
We assume $B_{0}$ is nonnegative and fixed.
The total mass of $B_{0}$, $\gamma := B_{0}(\Psi)$, is 
called the {\em mass parameter}.  We assume $\gamma$ is strictly positive 
and finite.
The density in \eq{beta_density}, with the choice of $B_{0}$ uniform 
over $[0,1]$, is illustrated in \fig{beta}.

\begin{figure}
\begin{center}
\renewcommand{\tabcolsep}{0cm}
\begin{tabular}{>{\centering}m{0.45\textwidth} >{\centering}m{0.45\textwidth}}
	\includegraphics[width=0.4\textwidth]{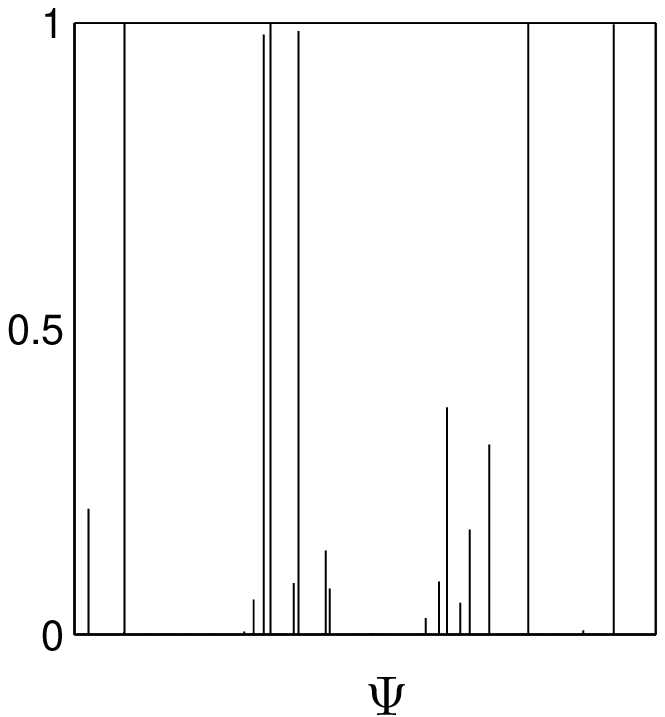}
	& \\
	\includegraphics[width=0.4\textwidth]{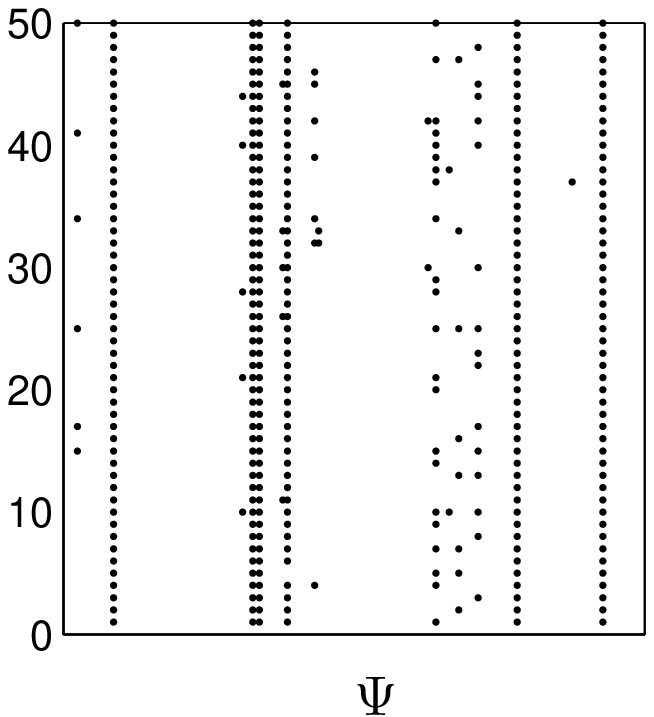}
	& \includegraphics[height=0.4\textwidth]{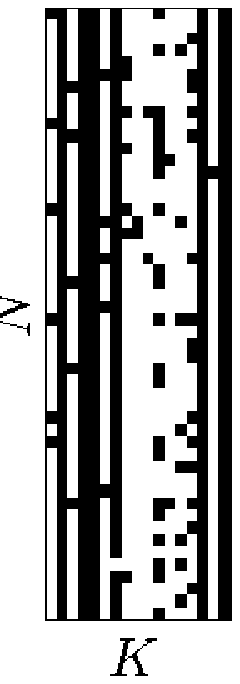}
\end{tabular}
\end{center}
\caption{\label{fig:bernoulli} {\em Upper left}: A draw $B$ from the beta process.  
{\em Lower left}: 50 draws from the Bernoulli process $BeP(B)$.  The vertical 
axis indexes the draw number among the 50 exchangeable draws. A point indicates 
a one at the corresponding location on the horizontal axis, $\psi \in \Psi$.
{\em Right}: We can form a matrix from the lower left plot by including 
only those $\psi$ values with a non-zero number of Bernoulli successes 
among the 50 draws from the Bernoulli process. Then, the number of columns 
$K$ is the number of such $\psi$, and the number of rows $N$ is the 
number of draws made.  A black square indicates a one at the corresponding 
matrix position; a white square indicates a zero.}
\end{figure}

The beta process can be viewed as providing an infinite collection of
coin-tossing probabilities.  Tossing these coins corresponds to a draw
from the \emph{Bernoulli process}, yielding an infinite binary vector
that we will treat as a latent feature vector.

More formally, a \emph{Bernoulli process} $Y \sim BeP(B)$ is a completely
random measure with potentially both fixed atomic and ordinary components.
In defining the Bernoulli process we consider only the case in which $B$ 
is discrete, i.e., of the form in \eq{crm_draw}, though not necessarily 
a beta process draw or even random for the moment.  Then $Y$ has only a 
fixed atomic component and has the form
\begin{equation}
	\label{eq:bep_discr}
	Y = \sum_{i=1}^{\infty} b_{i} \delta_{\psi_{i}},
\end{equation}
where $b_{i} \sim \bern(u_{i})$ for $u_{i}$ the corresponding atomic mass in the measure $B$.
We can see that $\mbe(Y|B) = B(\Psi)$ from the mean of the Bernoulli distribution,
so the number of non-zero points in any realization of the Bernoulli process is 
finite when $B$ is a finite measure.

We can link the beta process and $N$ Bernoulli process draws to generate a random feature
matrix $Z$.  To that end, first draw $B \sim \bp(\theta,B_{0})$ for fixed 
hyperparameters $\theta$ and $B_{0}$ and then draw $Y_{n} \stackrel{iid}{\sim} \bep(B)$ 
for $n \in \{1,\ldots,N\}$.  Note that since $B$ is discrete, 
each $Y_{n}$ will be discrete as in \eq{bep_discr}, with point masses 
only at the atoms $\{\psi_{i}\}$ of the beta process $B$.  Note also that
$\mbe B(\Psi) = \gamma < \infty$, so $B$ is a finite measure, and it 
follows that the number of non-zero point masses in any draw $Y_{n}$ 
from the Bernoulli process will be finite.  Therefore, the total number 
of non-zero point masses $K$ across $N$ such Bernoulli process draws is finite.

Now reorder the $\{\psi_{i}\}$ so that the first $K$ are exactly those 
locations where some Bernoulli process in $\{Y_{n}\}_{n=1}^{N}$ has a 
non-zero point mass.  We can form a matrix $Z \in \{0,1\}^{N \times K}$ 
as a function of the $\{Y_{n}\}_{n=1}^{N}$ by letting the $(n,k)$ 
entry equal one when $Y_{n}$ has a non-zero point mass at $\psi_{k}$ and zero 
otherwise.  If we wish to think of $Z$ as having an infinite number of 
columns, the remaining columns represent the point masses of the 
$\{Y_{n}\}_{n=1}^{N}$ at $\{\psi_{k}\}_{k > K}$, which we know to 
be zero by construction.  We refer to the overall procedure of drawing 
$Z$ according to, first, a beta process and then repeated Bernoulli 
process draws in this way as a \emph{beta-Bernoulli process}, and we 
write $Z \sim \bpbep(N,\gamma,\theta)$.  Note that we have
implicitly integrated out the $\{\psi_{k}\}$,
and the distribution of the matrix $Z$
depends on $B_{0}$ only through its total mass, $\gamma$.
As shown by~\citet{thibaux:2007:hierarchical},
this process yields the same distribution on row-exchangeable, infinite-column 
matrices as the Indian buffet process~\citep{griffiths:2006:infinite}, which
describes a stochastic process directly on (equivalence classes of) binary matrices.
That is, the Indian buffet process is obtained as an exchangeable distribution on 
binary matrices when the underlying beta process measure is integrated out.
This result is analogous to the derivation of the Chinese restaurant process as the 
exchangeable distribution on partitions obtained when the underlying Dirichlet 
process is integrated out.  The beta-Bernoulli process is illustrated in \fig{bernoulli}.

\section{Stick-breaking for the Dirichlet process} \label{sec:dp_stick}

\begin{figure}
\begin{center}
\resizebox{0.7\textwidth}{!}{\input{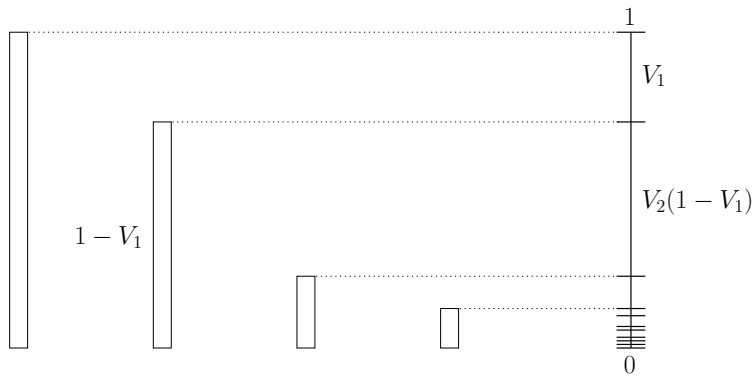}}
\end{center}
\caption{\label{fig:stick} A stick-breaking process
starts with the unit interval ({\em far left}). First, a
random fraction $V_{1}$ of the unit interval is broken
off; the remaining stick has length $1-V_{1}$
({\em middle left}). Next, a random fraction $V_{2}$
of the remaining stick is broken off, i.e., a fragment
of size $V_{2}(1-V_{1})$; the remaining stick has
length $(1-V_{1})(1-V_{2})$. This process proceeds
recursively and generates stick fragments
$V_{1},V_{2}(1-V_{1}),\ldots,V_{i}\prod_{j < i} (1-V_{j}), \ldots$.
These fragments form a random partition of the
unit interval ({\em far right}).}
\end{figure}

The stick-breaking representation of the Dirichlet 
process~\citep{mccloskey:1965:model,
patil:1977:diversity,sethuraman:1994:constructive}
provides a simple recursive
procedure for obtaining the weights $\{\pi_i\}$ in \eq{DP}.
This procedure provides an explicit representation of a draw $G$ from the 
Dirichlet process, one which can be usefully instantiated and updated in 
posterior inference algorithms~\citep{ishwaran:2001:gibbs,blei:2006:variational}. 
We begin this section by reviewing this stick-breaking construction 
as well as some of the extensions to this construction that yield 
power-law behavior.  We then turn to a consideration of stick-breaking 
and power laws in the setting of the beta process.

Stick-breaking is the process of recursively breaking off random fractions 
of the unit interval.  In particular, let $V_{1},V_{2},\ldots$ be some 
countable sequence of random variables, each with range $[0,1]$. Each 
$V_{i}$ represents the fraction of the remaining stick to break off at 
step $i$.  Thus, the first stick length generated by the stick-breaking 
process is $V_{1}$.  At this point, a fragment of length $1-V_{1}$ of the 
original stick remains.  Breaking off $V_{2}$ fraction of the remaining stick 
yields a second stick fragment of $V_{2}(1-V_{1})$.  This process iterates such that 
the stick length broken off at step $i$ is $V_{i} \prod_{j < i} (1-V_{j})$. 
The stick-breaking recursion is illustrated in \fig{stick}.

The Dirichlet process 
arises from the special case in which the $V_{i}$ are independent draws
from the $\tb(1,\theta)$ distribution~\citep{mccloskey:1965:model,
patil:1977:diversity,sethuraman:1994:constructive}.  Thus we have the following 
representation of a draw $G \sim \tdp(\theta, G_{0})$:
\begin{eqnarray}
	\nonumber
	G &=& \sum_{i=1}^{\infty} \left[ V_{i} \prod_{j=1}^{i-1} (1-V_{j}) \right] \delta_{\psi_{i}} \\
	\nonumber
	V_{i} &\stackrel{iid}{\sim}& \tb(1,\theta) \\
	\label{eq:dp_stick}
	\psi_{i} &\stackrel{iid}{\sim}& G_{0},
\end{eqnarray}
where $G_0$ is referred to as the \emph{base measure} and $\theta$
is referred to as the \emph{concentration parameter}.

\section{Power law behavior} \label{sec:power_law}

Consider the process of sampling a random measure $G$ from a Dirichlet 
process and subsequently drawing independently $N$ times from $G$.  
The number of unique atoms sampled according to this process will
grow as a function of $N$.  The growth associated with the Dirichlet
process is relatively slow, however, and when the Dirichlet process
is used as a prior in a clustering model one does not obtain the 
heavy-tailed behavior commonly referred to as a ``power law.''  In 
this section we first provide a brief exposition of the different kinds 
of power law that we might wish to obtain in a clustering model and 
discuss how these laws can be obtained via an extension of the 
stick-breaking representation.  We then discuss analogous laws for 
featural models.

\subsection{Power laws in clustering models} \label{sec:clustering_power}

First, we establish some notation.  Given a number $N$ of draws from 
a discrete random probability measure $G$ (where $G$ is not necessarily 
a draw from the Dirichlet process), let $(N_{1},N_{2},\ldots)$
denote the sequence of counts associated with the unique values
obtained among the $N$ draws, where we view these unique values
as ``clusters.''  Let
\begin{equation}
	\label{eq:frac_groups_j}
	K_{N,j} = \sum_{i=1}^{\infty} \mathbbm{1}(N_{i} = j),
\end{equation}
and let
\begin{equation}
	\label{eq:num_groups}
	K_{N} = \sum_{i=1}^{\infty} \mathbbm{1}(N_{i} > 0).
\end{equation}
That is, $K_{N,j}$ is the number of clusters that are drawn exactly
$j$ times, and $K_{N}$ is the total number of clusters.

There are two types of power-law behavior that a clustering model might 
exhibit.  First, there is the type of power law behavior reminiscent of 
Heaps' law~\citep{heaps:1978:information, gnedin:2007:notes}:
\begin{equation}
	\label{eq:heaps_law}
	K_{N} \stackrel{a.s.}{\sim} c N^{a}, \quad N \rightarrow \infty
\end{equation}
for some constants $c > 0, a \in (0,1)$.  Here, $\sim$ means that the 
limit of the ratio of the left-hand and right-hand side, when they are both 
real-valued and non-random, is one as the number of data points $N$ 
grows large.  We denote a power law in the form of \eq{heaps_law} as {\em Type I}.
Second, there is the type of power law behavior reminiscent 
of Zipf's law~\citep{zipf:1949:human,gnedin:2007:notes}:
\begin{equation}
	\label{eq:zipf_law}
	K_{N,j} \stackrel{a.s.}{\sim} \frac{a \Gamma(j-a)}{j! \Gamma(1-a)} c N^{a} \quad N \rightarrow \infty
\end{equation}
again for some constants $c > 0, a \in (0,1)$.
We refer to the power law in \eq{zipf_law} as {\em Type II}.

Sometimes in the case 
of \eq{zipf_law}, we are interested in the behavior in $j$; 
therefore we recall $j! = \Gamma(j+1)$ and
note the following fact about the $\Gamma$-function ratio
in \eq{zipf_law}~\citep[cf.][]{tricomi:1951:asymptotic}:
\begin{equation} \label{eq:asymp_j}
	\frac{\Gamma(j-a)}{\Gamma(j+1)} \sim j^{-1-a} \quad j \rightarrow \infty
\end{equation}
Again, we see behavior in the form of a power law at work.


Power-law behavior of Types I and II~\citep[and equivalent 
formulations; see][]{gnedin:2007:notes} has been observed in a variety of 
real-world clustering problems including, but not limited to: the number 
of species per plant genus, the in-degree or out-degree of a graph 
constructed from hyperlinks on the Internet, the number of people in cities, 
the number of words in documents, the number of papers published by 
scientists, and the amount each person earns in
income~\citep{mitzenmacher:2004:brief,goldwater:2006:interpolating}.
Bayesians modeling these situations will prefer 
a prior that reflects this distributional attribute.

While the Dirichlet process exhibits neither type of power-law behavior, 
the \emph{Pitman-Yor process} yields both kinds of
power law~\citep{pitman:1997:two, 
goldwater:2006:interpolating} though we note that in this case $c$
is a random variable (still with no dependence on $N$ or $j$).
The Pitman-Yor process, denoted 
$G \sim \py(\theta, \alpha, G_{0})$, is defined via the following
stick-breaking representation:
\begin{eqnarray}
	\nonumber
	G &=& \sum_{i=1}^{\infty} \left[ V_{i} \prod_{j=1}^{i-1} (1-V_{j}) \right] \delta_{\psi_{i}} \\
	\nonumber
	V_{i} &\stackrel{indep}{\sim}& \tb(1-\alpha,\theta+i\alpha) \\
	\label{eq:py_stick}
	\psi_{i} &\stackrel{iid}{\sim}& G_{0},
\end{eqnarray}
where $\alpha$ is known as a \emph{discount parameter}. The case $\alpha=0$ 
returns the Dirichlet process (cf.\ \eq{dp_stick}).

Note that in both the Dirichlet process and Pitman-Yor process, 
the weights $\{V_{i} \prod_{j=1}^{i-1} (1-V_{j})\}$ are the weights of 
the process in size-biased order~\citep{pitman:2006:combinatorial}. 
In the Pitman-Yor case, the $\{V_{i}\}$ are no longer identically distributed.

\subsection{Power laws in featural models} \label{sec:power_mult}

The beta-Bernoulli process provides a specific kind of feature-based 
representation of entities.  In this section we study general featural 
models and consider the power laws that might arise for such models.

In the clustering framework, we considered $N$ draws from a process that
put exactly one mass of size one on some value in $\Psi$
and mass zero elsewhere.
In the featural framework we consider $N$ draws from a process that places
some non-negative integer number of
masses, each of size one, on an almost surely finite set of values in $\Psi$
and mass zero elsewhere.
As $N_{i}$ was the sum of masses at a point labeled $\psi_{i} \in \Psi$
in the clustering framework, so do we now let $N_{i}$ be the sum of masses
at a point labeled $\psi_{i} \in \Psi$.
We use the same notation as in \mysec{clustering_power}, but now we note 
that the counts $N_i$ no longer sum to $N$ in general.

In the case of featural models, we can still talk about Type I and II 
power laws, both of which have the same interpretation as in the case
of clustering models.  In the featural case, however, it is also 
possible to consider a third type of power law.  If we let $k_{n}$ 
denote the number of features present in the $n$th draw, we say 
that $k_{n}$ shows power law behavior if
$$
	\mbp(k_{n} > M) \sim c M^{-a}
$$
for positive constants $c$ and $a$.  We call this last type of power 
law {\em Type III}.

\section{Stick-breaking for the beta process} \label{sec:bp_stick}

The weights $\{q_i\}$ for the beta process can be derived by a variety of
procedures, including size-biased sampling~\citep{thibaux:2007:hierarchical}
and inverse L\'evy measure~\citep{wolpert:2004:reflecting,teh:2007:stick}.  
The procedures that are closest in spirit to the stick-breaking representation for 
the Dirichlet process are those due to~\citet{paisley:2010:stick} and~\citet{teh:2007:stick}.
Our point of departure is the former, which has the following form:
\begin{eqnarray}
	\nonumber
	B &=& \sum_{i=1}^{\infty} \sum_{j=1}^{C_{i}} V_{i,j}^{(i)} \prod_{l=1}^{i-1} (1-V_{i,j}^{(l)}) \delta_{\psi_{i,j}} \\
	\nonumber
	C_{i} &\stackrel{iid}{\sim}& \pois(\gamma) \\
	\nonumber
	V_{i,j}^{(l)} &\stackrel{iid}{\sim}& \tb(1,\theta) \\
	\label{eq:stick-breaking}
	\psi_{i,j} &\stackrel{iid}{\sim}& \frac{1}{\gamma} B_{0}.
\end{eqnarray}
This representation is analogous to the stick-breaking representation of the 
Dirichlet process in that it represents a draw from the beta process as a sum 
over independently drawn atoms, with the weights obtained by a recursive procedure.
However, it is worth noting that for every $(i,j)$ tuple subscript for $V_{i,j}^{(l)}$,
a different stick exists and
is broken across the superscript $l$. Thus, there are no special additive properties 
across weights in the sum in \eq{stick-breaking}; by contrast, the weights in
\eq{py_stick} sum to one almost surely.

The generalization of the one-parameter Dirichlet process to the two-parameter 
Pitman-Yor process suggests that we might consider generalizing the stick-breaking 
representation of the beta process in \eq{stick-breaking} as follows:
\begin{eqnarray}
	\nonumber
	B &=& \sum_{i=1}^{\infty} \sum_{j=1}^{C_{i}} V_{i,j}^{(i)} \prod_{l=1}^{i-1} (1-V_{i,j}^{(l)}) \delta_{\psi_{i,j}} \\
	\nonumber
	C_{i} &\stackrel{iid}{\sim}& \pois(\gamma) \\
	\nonumber
	V_{i,j}^{(l)} &\stackrel{indep}{\sim}& \tb(1-\alpha,\theta+i\alpha) \\
	\label{eq:stick-breaking_two}
	\psi_{i,j} &\stackrel{iid}{\sim}& \frac{1}{\gamma} B_{0}.
\end{eqnarray}
In \mysec{power_proofs} we will show that introducing the additional 
parameter $\alpha$ indeed yields Type I and II power law behavior (but not Type III). 

In the remainder of this section we present a proof that these stick-breaking
representations arise from the beta process.  In contradistinction to the proof of \eq{stick-breaking}
by~\citet{paisley:2010:stick}, which used a limiting process defined on sequences
of finite binary matrices, our approach makes a direct connection to the 
Poisson process characterization of the beta process.  Our proof has
several virtues: (1) it relies on no asymptotic arguments and instead comes 
entirely from the Poisson process representation; (2) it is, as a result, simpler 
and shorter; and (3) it demonstrates clearly the ease of incorporating a third parameter 
analogous to the discount parameter of the Pitman-Yor process and thereby 
provides a strong motivation for the extended stick-breaking
representation in \eq{stick-breaking_two}.

Aiming toward the general stick-breaking representation in \eq{stick-breaking_two},
we begin by defining a three-parameter generalization
of the beta process.\footnote{See also \citet{teh:2009:indian} or \citet{kim:2001:posterior}, with
$\theta(t) \equiv 1-\alpha, \beta(t) \equiv \theta + \alpha$, where the left-hand sides are in the notation of \citet{kim:2001:posterior}.}
We say that $B \sim \bp(\theta,\alpha,B_{0})$, 
where we call $\alpha$ a {\em discount parameter},
if, for $\psi \in \Psi, u \in [0,1])$, we have
\begin{equation}
	\label{eq:beta_density_three}
	\nu_{\bp}(d\psi, du)
		= \frac{\Gamma(1+\theta)}{\Gamma(1-\alpha)\Gamma(\theta+\alpha)} u^{-1-\alpha} (1-u)^{\theta+\alpha-1} \; du \; B_{0}(d\psi).
\end{equation}
It is straightforward to show that this three-parameter density has similar properties 
to that of the two-parameter beta process. For instance, choosing $\alpha \in (0,1)$ and 
$\theta > -\alpha$ is necessary for the beta process to have finite total mass almost
surely; in this case,
\begin{equation}
	\label{eq:finite_mass_beta}
	\int_{\Psi \times \mathbb{R}_{+}} u \; \nu_{\bp}(d\psi, du)
		= \frac{\Gamma(1-\alpha)\Gamma(\theta+\alpha)}{\Gamma(1+\theta)} < \infty.
\end{equation}

We now turn to the main result of this section.
\begin{proposition} $B$ can be represented according to the process 
described in \eq{stick-breaking_two} if and only if 
$B \sim \bp(\theta, \alpha, B_{0})$. 
\end{proposition}

\begin{proof}
First note that the points in the set
$$
	P_{1} := \left\{(\psi_{1,1}, V_{1,1}^{(1)}), (\psi_{1,2}, V_{1,2}^{(1)}), \ldots, (\psi_{1,C_{1}}, V_{1,C_{1}}^{(1)}) \right\}
$$
are by construction independent and identically distributed conditioned on $C_{1}$. 
Since $C_{1}$ is Poisson-distributed, $P_{1}$ is a Poisson point process. 
The same logic gives that in general, for
$$
	P_{i} := \left\{
				\left( \psi_{i,1}, V_{i,1}^{(i)} \prod_{l=1}^{i-1} (1-V_{i,1}^{(l)}) \right),
				\ldots,
				\left( \psi_{i,C_{i}}, V_{i,C_{i}}^{(i)} \prod_{l=1}^{i-1} (1-V_{i,C_{i}}^{(l)}) \right)
			\right\},
$$
$P_{i}$ is a Poisson point process.

Next, define
$$
	P := \bigcup_{i=1}^{\infty} P_{i}
$$
As the countable union of Poisson processes with finite rate measures, $P$ is itself a Poisson point process.

Notice that we can write $B$ as the completely random measure $B = \sum_{(\psi,U) \in P} U \delta_{\psi}$. Also, for any $B' \sim \bp(\theta, d, B_{0})$, we can write $B' = \sum_{(\psi',U') \in \Pi} U' \delta_{\psi'}$,
where $\Pi$ is Poisson point process with rate measure $\nu_{\bp} = B_{0} \times \mu_{\bp}$, and $\mu_{\bp}$ is a $\sigma$-finite measure with density
\begin{equation} \label{eq:beta_density_three_marginal}
	\frac{\Gamma(1+\theta)}{\Gamma(1-\alpha)\Gamma(\theta+\alpha)} u^{-1-\alpha}(1-u)^{\theta+\alpha-1} \; du.
\end{equation}
Therefore, to show that $B$ has the same distribution as $B'$, it is enough to show that $P$ and $\Pi$ have the same rate measures.

To that end, let $\nu$ denote the rate measure of $P$:
\begin{align}
	\nonumber
	\nu(A \times \tilde{A})
		\nonumber
		&= \mbe \#\{(\psi_{i}, U_{i}) \in A \times \tilde{A})\} \\
		\nonumber
		&= \frac{1}{\gamma} B_{0}(A) \cdot \mbe \sum_{i=1}^{\infty} \sum_{j=1}^{C_{i}} \mathbbm{1}\{ V_{ij}^{(i)} \prod_{l=1}^{i-1}(1-V_{ij}^{(l)}) \in \tilde{A}\} \\
		\label{eq:nu_sum}
		&= \frac{1}{\gamma} B_{0}(A) \cdot \sum_{i=1}^{\infty} \mbe \sum_{j=1}^{C_{i}} \mathbbm{1}\{ V_{ij}^{(i)} \prod_{l=1}^{i-1}(1-V_{ij}^{(l)}) \in \tilde{A}\},
\end{align}
where the last line follows by monotone convergence.
Each term in the outer sum can be further decomposed as
\begin{align}
	\nonumber
	\mbe \sum_{j=1}^{C_{i}} \mathbbm{1}\{ V_{ij}^{(i)} \prod_{l=1}^{i-1}(1-V_{ij}^{(l)}) \in \tilde{A}\} 
		&= \mbe\left[ \mbe\left[ \left. \sum_{j=1}^{C_{i}} \mathbbm{1}\{ V_{ij}^{(i)} \prod_{l=1}^{i-1}(1-V_{ij}^{(l)}) \in \tilde{A}\} \right| C_{i} \right] \right] \\
		\nonumber
		&= \mbe\left[ C_{i} \right] \mbe \left[ \mathbbm{1}\{ V_{i1}^{(i)} \prod_{l=1}^{i-1}(1-V_{i1}^{(l)}) \in \tilde{A}\} \right] \\
		\nonumber
		& \textrm{since the $V_{ij}^{(l)}$ are iid across $j$ and independent of $C_{i}$} \\
		\label{eq:sum_term}
		&= \gamma \; \mbe \mathbbm{1}\{ V_{i} \prod_{l=1}^{i-1}(1-V_{l}) \in \tilde{A} \} \\
		\nonumber
		& \textrm{for $V_{i} \stackrel{indep}{\sim} \tb(1-\alpha,\theta+i\alpha)$},
\end{align}
where the last equality follows since the choice of $\{V_{i}\}$ gives $V_{i} \prod_{l=1}^{i-1}(1-V_{l}) \stackrel{d}{=} V_{i1}^{(i)} \prod_{l=1}^{i-1}(1-V_{i1}^{(l)})$.

Substituting \eq{sum_term} back into \eq{nu_sum}, canceling $\gamma$ factors, and applying monotone convergence again yields
\begin{align*}
	\nu(A \times \tilde{A})
		&= B_{0}(A) \cdot \mbe \sum_{i=1}^{\infty} \mathbbm{1}\{ V_{i} \prod_{l=1}^{i-1}(1-V_{l}) \in \tilde{A} \}.
\end{align*}

We note that both of the measures $\nu$ and $\nu_{\bp}$ factorize:
\begin{eqnarray*}
	\nu(A \times \tilde{A})
		&=& B_{0}(A) \cdot \mbe \sum_{i=1}^{\infty} \mathbbm{1}\{ V'_{i} \prod_{l=1}^{i-1}(1-V'_{l}) \in \tilde{A}\} \\
	\nu_{BP}(A \times \tilde{A})
		&=& B_{0}(A) \mu_{\bp}(\tilde{A}),
\end{eqnarray*}
so it is enough to show that $\mu = \mu_{\bp}$ for the measure $\mu$ defined by
\begin{equation}
	\label{eq:meas_set}
	\mu(\tilde{A}) := \mbe \sum_{i=1}^{\infty} \mathbbm{1}\{ V_{i} \prod_{l=1}^{i-1}(1-V_{l}) \in \tilde{A}\}.
\end{equation}

At this point and later in proving \prop{abel_taub}, we will make use of part of Campbell's theorem, which we copy here from~\citet{kingman:1993:poisson} for completeness.
\begin{theorem}[Part of Campbell's Theorem] \label{thm:campbell}
	Let $\Pi$ be a Poisson process on $S$ with rate measure $\mu$, and let $f:S \rightarrow \mathbb{R}$ be measurable. If $\int_{S} \min(|f(x)|,1) \; \mu(dx) < \infty$, then
	\begin{equation} \label{eq:campbell}
		\mbe\left[ \sum_{X \in \Pi} f(X) \right] = \int_{S} f(x) \; \mu(dx).
	\end{equation}
\end{theorem}

Now let $\tilde{U}$ be a size-biased pick from $\{V_{i} \prod_{l=1}^{i-1}(1-V_{l})\}_{i=1}^{\infty}$. By construction, for any bounded, measurable function $g$, we have
$$
	\mbe\left[ g(\tilde{U}) | \{V_{i}\} \right]
		= \sum_{i=1}^{\infty} V_{i} \prod_{l=1}^{i-1}(1-V_{l}) \cdot g(V_{i} \prod_{l=1}^{i-1}(1-V_{l})).
$$
Taking expectations yields
$$
	\mbe g(\tilde{U})
	 	= \mbe \left[ \sum_{i=1}^{\infty} V_{i} \prod_{l=1}^{i-1}(1-V_{l}) g(V_{i} \prod_{l=1}^{i-1}(1-V_{l})) \right]
		= \int u g(u) \mu(du),
$$
where the final equality follows by Campbell's theorem with the choice $f(u) = u g(u)$.
Since this result holds for all bounded, measurable $g$, we have that
\begin{equation}
	\label{eq:size_bias_dist}
	\mbp(\tilde{U} \in du)
		= u \mu(du).
\end{equation}

Finally, we note that, by \eq{meas_set}, $\tilde{U}$ is a size-biased sample 
from probabilities generated by stick-breaking with proportions $\{\tb(1-\alpha,\theta+i\alpha)\}$.
Such a sample is then distributed $\tb(1-\alpha,\theta+\alpha)$ since, as mentioned above, the Pitman-Yor stick-breaking construction gives the size-biased frequencies in order. So, rearranging \eq{size_bias_dist}, we can write
\begin{eqnarray*}
	\mu(du)
		&=& u^{-1} \mbp(\tilde{U} \in du) \\
		&=& u^{-1} \frac{\Gamma(1+\theta)}{\Gamma(1-\alpha)\Gamma(\theta+\alpha)} u^{(1-\alpha)-1} (1-u)^{(\theta+\alpha)-1} \\
		&& \textrm{using the $\tb(1-\alpha,\theta+\alpha)$ density} \\
		&=& \mu_{\bp}(du),
\end{eqnarray*}
as was to be shown.
\end{proof}

\section{Power law derivations} \label{sec:power_proofs}

By linking the three-parameter stick-breaking representation to the power-law 
beta process in \eq{beta_density_three}, we can use the results of the following 
section to conclude that the feature assignments in the three-parameter model 
follow both Type I and Type II power laws and that they do not follow a Type 
III power law (\mysec{power_mult}).  We note that~\citet{teh:2009:indian} found 
big-O behavior for Types I and II in the three-parameter Beta and a Poisson 
distribution for the Type III distribution.  We can strengthen these results 
to obtain exact asymptotic behavior with constants in the first two cases 
and also conclude that Type III power laws can never hold in the featural 
framework whenever the sum of the feature probabilities is almost surely finite, 
an assumption that would appear to be a necessary component of any physically 
realistic model. 

\subsection{Type I and II power laws} \label{sec:types_i_ii_proofs}

Our subsequent derivation expands upon the work of~\citet{gnedin:2007:notes}. 
In that paper, the main thrust of the argument applies to the case in which
the feature probabilities are fixed rather than random. 
In what follows, we obtain power laws of Type I and II in the case in
which the feature probabilities are random, in particular when the probabilities 
are generated from a Poisson process.  We will see that this last assumption 
becomes convenient in the course of the proof.  Finally, we apply our results 
to the specific example of the beta-Bernoulli process.

Recall that we defined $K_{N}$, the number of represented clusters
in the first $N$ data points, and $K_{N,j}$, the number of clusters 
represented $j$ times in the first $N$ data points,
in \eqs{num_groups} and \eqss{frac_groups_j}, respectively.
In \mysec{power_mult}, we noted that same definitions in
\eqs{num_groups} and \eqss{frac_groups_j}
hold for featural models if we now let $N_{i}$
denote the number of data points at time $N$ in which feature $i$
is represented. In terms of the Bernoulli process, $N_{i}$ would
be the number of Bernoulli process draws, out of $N$, where the $i$th atom has
unit (i.e., nonzero) weight. It need not be the case that the $N_{i}$
sum to $N$.

Working directly to find power laws in $K_{N}$ and $K_{N,j}$
as $N$ increases is challenging 
in part due to $N$ being an integer. A standard technique to surmount this 
difficulty is called {\em Poissonization}. In Poissonizing $K_{N}$ and 
$K_{N,j}$, we consider new functions $K(t)$ and 
$K_{j}(t)$ where the argument $t$ is continuous, in contrast to the integer argument $N$.
We will define $K(t)$ and $K_{j}(t)$
such that $K(N)$ and $K_{j}(N)$ have the same asymptotic behavior 
as $K_{N}$ and $K_{N,j}$, respectively.

In particular, our derivation of the asymptotic behavior of $K_{N}$ and $K_{N,j}$ will
consist of three parts and will involve working extensively with the mean feature counts
$$
	\Phi_{N} := \mbe[K_{N}] \quad \textrm{and} \quad \Phi_{N,j} := \mbe[K_{N,j}] \quad (j > 1)
$$
with $N \in \{1,2,\ldots\}$ and the Poissonized mean feature counts
$$
	\Phi(t) := \mbe[K(t)] \quad \textrm{and} \quad \Phi_{j}(t) := \mbe[K_{j}(t)] \quad (j > 1).
$$
with $t > 0$.
First, we will take advantage of Poissonization to find power laws in $\Phi(t)$
and $\Phi_{j}(t)$ as $t \rightarrow \infty$ (\prop{abel_taub}). Then, in order
to relate these results back to the original process, we will show that $\Phi_{N}$ and $\Phi(N)$
have the same asymptotic behavior and also that $\Phi_{N,j}$ and $\Phi_{j}(N)$ have the
same asymptotic behavior (\lem{poissonization}). Finally, to obtain results
for the random process values $K_{N}$
and $K_{N,j}$, we will conclude by showing that $K_{N}$ almost surely
has the same asymptotic behavior as $\Phi_{N}$ and that $\sum_{k < j} K_{N,k}$
almost surely has the same asymptotic behavior as $\sum_{k < j} \Phi_{N,k}$
(\prop{as_mean}).

To obtain power laws for the Poissonized process,
we must begin by defining $K(t)$ and $K_{j}(t)$.
To do so, we will construct Poisson processes on the 
positive half-line, one for each feature. $K(t)$ will be the number of such Poisson 
processes with points in the interval $[0,t]$; similarly, $K_{j}(t)$ 
will be the number of Poisson processes with $j$ points in the interval 
$[0,t]$. This construction is illustrated in \fig{poissonization}. It remains to 
specify the rates of these Poisson processes.

\begin{figure}
	\begin{center}
	\includegraphics[height=0.5\textwidth]{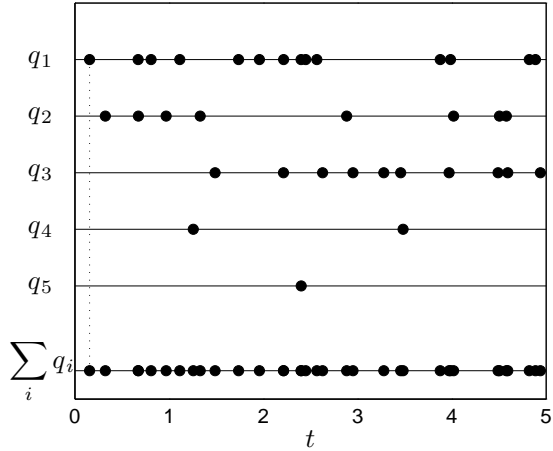}
	\end{center}
	\caption{\label{fig:poissonization}
	The first five sets of points, starting from the top of the figure, illustrate Poisson processes on the positive half-line in the range $t \in (0,5)$ with respective rates $q_{1},\ldots, q_{5}$. The bottom set of points illustrates the union of all points from the preceding Poisson point processes and is, therefore, itself a Poisson process with rate $\sum_{i} q_{i}$. In this example, we have for instance that $K(1) = 2$, $K(4)=5$, and $K_{2}(4) = 1$.}
\end{figure}

Let $(q_{1},q_{2},\ldots)$ be a countably infinite vector of
feature probabilities. We begin by putting minimal restrictions
on the $q_{i}$. We assume that they are strictly positive,
decreasing real numbers. They need not necessarily sum
to one, and they may be random.
Indeed, we will eventually consider the case where
the $q_{i}$ are the (random) atom weights of a beta process,
and then
we will have $\sum_{i} q_{i} \ne 1$ with probability one.

Let $\Pi_{i}$ be a standard Poisson process on the positive
real line generated with rate $q_{i}$ (see, e.g., the top five lines
in \fig{poissonization}). Then $\Pi := \bigcup_{i} \Pi_{i}$
is a standard Poisson process on the positive real line with rate
$\sum_{i} q_{i}$ (see, e.g., the lowermost line in \fig{poissonization}),
where we henceforth assume
$\sum_{i} q_{i} < \infty$ a.s. 

Finally, as mentioned above,
we define $K(t)$ to be the number of Poisson processes
$\Pi_{i}$ with any points in $[0,t]$:
$$
	K(t) := \sum_{i} \mathbbm{1}\{ | \Pi_{i} \cap [0,t] | > 0 \}.
$$
And we define $K_{j}(t)$ to be the number of Poisson processes
$\Pi_{i}$ with exactly $j$ points in $[0,t]$:
$$
	K_{j}(t) := \sum_{i} \mathbbm{1}\{ |\Pi_{i} \cap [0,t] | = j\}.
$$
These definitions are very similar to the definitions of $K_{N}$ and
$K_{N,j}$ in \eqs{num_groups} and \eqss{frac_groups_j}, respectively.
The principal difference is that the $K_{N}$ are incremented only at
integer $N$
whereas the $K(t)$ can have jumps at any $t \in \mathbb{R}_{+}$.
The same observation holds for the $K_{N,j}$ and $K_{j}(t)$.

In addition to Poissonizing $K_{N}$ and $K_{N,j}$ to define $K(t)$ and $K_{j}(t)$,
we will also find it convenient to assume that the $\{q_{i}\}$ themselves are derived
from a Poisson process with rate measure $\nu$. We note that
Poissonizing from a discrete index $N$ to a continuous time index $t$
is an approximation and separate from our assumption that the
$\{q_{i}\}$ are generated from a Poisson process though both are
fundamentally tied to the ease of working with Poisson processes.

We are now able to write out the mean feature 
counts in both the Poissonized and original cases. First, the 
Poissonized definitions
of $\Phi$ and $K$ allow us to write
$$
	\Phi(t) := \mbe[K(t)] = \mbe[ \mbe[K(t) | q] ] = \mbe[ \mbe[ \sum_{i} \mathbbm{1}\{|\Pi_{i} \cap [0,t] | > 0\} | q] ].
$$
With a similar approach for $\Phi_{j}(t)$, we find
$$
	\Phi(t) = \mbe[\sum_{i} (1-e^{-tq_{i}})], \quad \Phi_{j}(t) = \mbe[ \sum_{i} \frac{(tq_{i})^{j}}{j!} e^{-tq_{i}} ].
$$
With the assumption that the $\{q_{i}\}$ are drawn from a 
Poisson process with measure measure $\nu$, we can
apply Campbell's theorem (\thm{campbell}) to both
the original and Poissonized versions of the process to derive the final
equality in each of the following lines
\begin{align}
	\label{eq:phi_t}
	\Phi(t) &= \mbe[ \sum_{i} (1-e^{-tq_{i}}) ] = \int_{0}^{1} (1-e^{-tx}) \; \nu(dx) \\
	\label{eq:phi_N}
	\Phi_{N} &= \mbe[ \sum_{i} (1-(1-q_{i})^{N}) ] = \int_{0}^{1} (1-(1-x)^{N}) \; \nu(dx) \\
	\label{eq:phi_tj}
	\Phi_{j}(t) &= \mbe[ \sum_{i} \frac{(tq_{i})^{j}}{j!} e^{-tq_{i}} ] = \frac{t^{j}}{j!} \int_{0}^{1} x^{j} e^{-tx} \; \nu(dx) \\
	\label{eq:phi_Nj}
	\Phi_{N,j} &= \binom{N}{j} \mbe[ \sum_{i} q_{i}^{j} (1-q_{i})^{N-j} ] = \binom{N}{j} \int_{0}^{1} x^{j} (1-x)^{N-j} \; \nu(dx).
\end{align}

Now we establish our first result, which gives a power law in
$\Phi(t)$ and $\Phi_{j}(t)$ when the Poisson process rate measure $\nu$
has corresponding power law properties.

\begin{proposition} \label{prop:abel_taub} Asymptotic behavior of the integral
of $\nu$ of the following form
\begin{equation} \label{eq:vec_nu}
	\nu_{1}[0,x] := \int_{0}^{x} u \; \nu(du) \sim \frac{\alpha}{1-\alpha} x^{1-\alpha} l(1/x),
		\quad x \rightarrow 0
\end{equation}
where $l$ is a regularly varying function and $\alpha \in (0,1)$ implies
\begin{eqnarray*}
	\Phi(t) &\sim& \Gamma(1-\alpha) t^{\alpha} l(t),
		\quad t \rightarrow \infty \\
	\Phi_{j}(t) &\sim& \frac{\alpha \Gamma(j-\alpha)}{j!} t^{\alpha} l(t),
		\quad t \rightarrow \infty
		\quad (j > 1).
\end{eqnarray*}
\end{proposition}

\begin{proof}
The key to this result is in the repeated use of Abelian or Tauberian theorems.
Let $A$ be a map $A: F \rightarrow G$ from one function space to another: e.g.,
an integral or a Laplace transform. For $f \in F$, an Abelian theorem gives us the
asymptotic behavior of $A(f)$ from the asymptotic behavior of $f$, and a Tauberian
theorem gives us the asymptotic behavior of $f$ from that of $A(f)$. 

First, integrating by parts yields
$$
	\nu_{1}[0,x]  = -x \bar{\nu}(x) + \int_{0}^{x} \bar{\nu}(u) \; du, \quad \bar{\nu}(x) := \int_{x}^{\infty} \nu(u) \; du,
$$
so the stated asymptotic behavior in $\nu_{1}$ yields
$\bar{\nu}(x) \sim l(1/x) x^{-\alpha} (x \rightarrow 0)$ by
a Tauberian theorem~\citep{feller:1966:introduction,gnedin:2007:notes}
where the map $A$ is an integral.

Second, another integration by parts yields
$$
	\Phi(t) = t \int_{0}^{\infty} e^{-tx} \bar{\nu}(x) \; dx.
$$
The desired asymptotic behavior in $\Phi$ follows
from the asymptotic behavior in $\bar{\nu}$ and an
Abelian theorem~\citep{feller:1966:introduction,gnedin:2007:notes}
where the map $A$ is a 
Laplace transform. The result for $\Phi_{j}(t)$ 
follows from a similar argument when we note 
that repeated integration by parts of \eq{phi_tj} 
also yields a Laplace transform.
\end{proof}

The importance of assuming that the $q_{i}$ are distributed according
to a Poisson process is that this assumption allowed us to
write $\Phi$ as an integral and thereby make use of classic Abelian
and Tauberian theorems. The importance of Poissonizing the
processes $K_{j}$ and $K_{N,j}$ is that we can write their
means as in \eqs{phi_t} and \eqss{phi_tj}, which are---up
to integration by parts---in the form of Laplace transforms.

\prop{abel_taub} is the most significant link in the chain of results 
needed to show asymptotic behavior of the feature counts $K_{N}$ and
$K_{N,j}$ in that it relates power laws in the known feature probability
rate measure $\nu$ to power laws in the mean behavior of the
Poissonized version of these processes. It remains to show this mean
behavior translates back to $K_{N}$ and $K_{N,j}$, first by relating
the means of the original and Poissonized processes and then by
relating the means to the almost sure behavior of the counts.
The next two lemmas address the former concern.
Together they establish that the mean feature counts
$\Phi_{N}$ and  $\Phi_{N,j}$
have the same asymptotic behavior as the corresponding
Poissonized mean feature counts
$\Phi(N)$ and $\Phi_{j}(N)$.

\begin{lemma} \label{lem:basic_asymp}
Let $\nu$ be $\sigma$-finite with $\int_{0}^{\infty} \nu(du) = \infty$
and $\int_{0}^{\infty} u \; \nu(du) < \infty$.
Then the number of represented features has unbounded growth almost surely. 
The expected number of represented features has unbounded growth, 
and the expected number of features has sublinear growth. That is,
$$
	K(t) \uparrow \infty \textrm{ a.s}., \quad \Phi(t) \uparrow \infty, \quad \Phi(t) \ll t.
$$
\end{lemma}

\begin{proof}
As in~\citet{gnedin:2007:notes}, the first statement follows from the 
fact that $q$ is countably infinite and each $q_{i}$ is strictly positive. 
The second statement follows from monotone convergence. The final statement 
is a consequence of $\sum_{i} q_{i} < \infty$ a.s.
\end{proof}

\begin{lemma} \label{lem:poissonization}
Suppose the $\{q_{i}\}$ are generated according to a Poisson process
with rate measure as in \lem{basic_asymp}.  Then, for $N \rightarrow \infty$, 
\begin{align*}
	|\Phi_{N} - \Phi(N)| &< \frac{2}{N} \Phi_{2}(N) \rightarrow 0 \\
	|\Phi_{N,j} - \Phi_{j}(N)| &< \frac{c_{j}}{N} \max\{\Phi_{j}(N), \Phi_{j+2}(N)\} \rightarrow 0.
\end{align*}
for some constants $c_{j}$.
\end{lemma}

\begin{proof}
The proof is the same as that of Lemma 1 of~\citet{gnedin:2007:notes}. Establishing the inequalities results from algebraic manipulations. The convergence to zero is a consequence of \lem{basic_asymp}.
\end{proof}

Finally, before considering the specific case of the three-parameter beta 
process, we wish to show that power laws in the means $\Phi_{N}$ and 
$\Phi_{N,j}$ extend to almost sure power laws in the number of represented features.

\begin{proposition} \label{prop:as_mean}
Suppose the $\{q_{i}\}$ are generated from a Poisson process
with rate measure as in \lem{basic_asymp}. 
For $N \rightarrow \infty$,
$$
	K_{N} \stackrel{a.s.}{\sim} \Phi_{N}, \quad \sum_{k < j} K_{N,k} \stackrel{a.s.}{\sim} \sum_{k < j} \Phi_{N,k}.
$$
\end{proposition}

\begin{proof}
We wish to show that $K_{N} / \Phi_{N} \stackrel{a.s.}{\rightarrow} 1$ as 
$N \rightarrow \infty$. By Borel-Cantelli, it is enough to show that, 
for any $\epsilon > 0$,
$$
	\sum_{N} \mbp\left( \left| \frac{K_{N}}{\Phi_{N}} - 1 \right| > \epsilon \right) < \infty.
$$
To that end, note
$$
	\mbp\left( \left| K_{N} - \Phi_{N} \right| > \epsilon \Phi_{N} \right)
		\le \mbp\left( \Phi_{N} > \epsilon \Phi_{N} + K_{N} \right)
			+ \mbp\left( K_{N} > \epsilon \Phi_{N} + \Phi_{N} \right).
$$
The note after Theorem 4 in~\citet{freedman:1973:another} gives that
\begin{eqnarray*}
	\mbp\left( \Phi_{N} > \epsilon \Phi_{N} + K_{N} \right)
		&\le& \exp\left( - \epsilon^{2} \Phi_{N} \right) \\
	\mbp\left( K_{N} > \epsilon \Phi_{N} + \Phi_{N} \right)
		&\le& \exp\left( - \frac{\epsilon^{2}}{1+\epsilon} \Phi_{N} \right).
\end{eqnarray*}
So
\begin{eqnarray*}
	\mbp\left( \left| \frac{K_{N}}{\Phi_{N}} - 1 \right| > \epsilon \right)
		&\le& 2 \exp\left( - 2 \epsilon^{2} \Phi_{N} \right) \\
		&\le& c \exp\left( - 2 \epsilon^{2} N \right)
\end{eqnarray*}
for some constant $c$ and sufficiently large $N$ by \lems{basic_asymp} 
and \lemss{poissonization}. The last expression is summable in $N$, and 
Borel-Cantelli holds.

The proof that $\quad \sum_{k < j} K_{N,k} \stackrel{a.s.}{\sim} 
\sum_{k < j} \Phi_{N,j}$ follows the same argument.
\end{proof}

It remains to show that we obtain Type I and II power laws in our special 
case of the three-parameter beta process, which implies a particular rate 
measure $\nu$ in the Poisson process representation of the $\{q_{i}\}$. 
For the three-parameter beta process density in \eq{beta_density_three}, 
we have
\begin{eqnarray*}
	\nu_{1}[0,x] &=& \int_{\Psi \times (0,x]} u \; \nu_{BP}(d\psi,du) \\
		&=& \gamma \cdot \frac{ \Gamma(1+\theta)}{\Gamma(1-\alpha)\Gamma(\theta+\alpha)} \int_{0}^{x} u^{-\alpha} (1-u)^{\theta + \alpha - 1} \; du \\
		&\sim& \gamma \cdot \frac{ \Gamma(1+\theta)}{\Gamma(1-\alpha)\Gamma(\theta+\alpha)} \int_{0}^{x} u^{-\alpha} \; du, \quad x \downarrow 0 \\
		&=& \gamma \cdot \frac{ \Gamma(1+\theta)}{\Gamma(1-\alpha)\Gamma(\theta+\alpha)} \cdot \frac{1}{1-\alpha} x^{1-\alpha}.
\end{eqnarray*}
The final line is exactly the form required by \eq{vec_nu}
in \prop{abel_taub}, with $l(y)$ 
equal to the constant function of value
\begin{equation}
	\label{eq:C_val}
	C := \frac{\gamma}{\alpha} \cdot \frac{\Gamma(1+\theta)}{\Gamma(1-\alpha)\Gamma(\theta+\alpha)}.
\end{equation}

Then \prop{abel_taub} implies that the following power laws
hold for the mean of the Poissonized 
process:
\begin{eqnarray*}
	\Phi(t) &\stackrel{a.s.}{\sim}& \Gamma(1-\alpha) C t^{\alpha},
		\quad t \rightarrow \infty \\
	\Phi_{j}(t) &\stackrel{a.s.}{\sim}& \frac{\alpha \Gamma(j-\alpha)}{ j!} C t^{\alpha},
		\quad t \rightarrow \infty
		\quad (j > 1).
\end{eqnarray*}
\lem{poissonization} further yields
\begin{eqnarray*}
	\Phi_{N} &\stackrel{a.s.}{\sim}& \Gamma(1-\alpha) C N^{\alpha},
		\quad N \rightarrow \infty \\
	\Phi_{N,j} &\stackrel{a.s.}{\sim}& \frac{\alpha \Gamma(j-\alpha)}{ j!} C N^{\alpha},
		\quad N \rightarrow \infty
		\quad (j > 1),
\end{eqnarray*}
and finally \prop{as_mean} implies
\begin{eqnarray}
	\label{eq:power_total_beta}
	K_{N} &\stackrel{a.s.}{\sim}& \Gamma(1-\alpha) C N^{\alpha},
		\quad N \rightarrow \infty \\
	\label{eq:power_r_beta}
	K_{N,j} &\stackrel{a.s.}{\sim}& \frac{d \Gamma(j-\alpha)}{ j!} C N^{\alpha},
		\quad N \rightarrow \infty
		\quad (j > 1).
\end{eqnarray}
These are exactly the desired Type I and II power laws 
(\eqs{heaps_law} and \eqss{zipf_law}) for appropriate choices of the constants.

\subsection{Exponential decay in the number of features} \label{sec:exp_proof}

Next we consider a single data point and the number of features which
are expressed for that data point in the featural model. We prove 
results for the general case where the $i$th feature has probability 
$q_{i} \ge 0$ such that $\sum_{i} q_{i} < \infty$. Let $Z_{i}$ be 
a Bernoulli random variable with success probability $q_{i}$ and 
such that all the $Z_{i}$ are independent. Then $\mbe[\sum_{i} Z_{i}] 
= \sum_{i} q_{i} =: Q$. In this case, a
Chernoff bound~\citep{chernoff:1952:measure,hagerup:1990:guided} tells us 
that, for any $\delta > 0$, we have
$$
	\mbp[ \sum_{i} Z_{i} \ge (1+\delta) Q ]
		\le e^{\delta Q} (1+\delta)^{-(1+\delta) Q}.
$$
When $M$ is large enough such that $M > Q$, we can choose $\delta$ 
such that $(1+\delta) Q = M$. Then this inequality becomes
\begin{equation}
\label{eq:no_type_iii}
	\mbp[ \sum_{i} Z_{i} \ge M ]
		\le e^{M - Q} Q^{M} M^{-M} \quad \textrm{for $M > Q$}.
\end{equation}

We see from \eq{no_type_iii}
that the number of features $\sum_{i} Z_{i}$ that 
are expressed for a data point exhibits super-exponential tail 
decay and therefore cannot have a power law probability distribution 
when the sum of feature probabilities $\sum_{i} q_{i}$ is finite.
For comparison, let $Z \sim \pois(Q)$. Then~\citep{franceschetti:2007:closing}
$$
	\mbp[ Z \ge M ]
		\le e^{M - Q} Q^{M} M^{-M} \quad \textrm{for $M > Q$},
$$
the same tail bound as in \eq{no_type_iii}.

To apply the tail-behavior result of \eq{no_type_iii}
to the beta process (with two or three parameters),
we note that the total feature probability mass is finite by \eq{finite_mass_beta}.
Since the same set of feature probabilities is used in all subsequent Bernoulli
process draws for the beta-Bernoulli process, the result holds.

\section{Simulation} \label{sec:simulation}

To illustrate the three types of power laws discussed above,
we simulated beta process atom weights under three different
choices of the discount parameter $\alpha$, namely $\alpha=0$
(the classic, two-parameter beta process), $\alpha=0.3$, and
$\alpha=0.6$. In all three simulations, the remaining beta process
parameters were kept constant at total mass parameter
value $\gamma = 3$ and concentration parameter value $\theta = 1$.

The simulations were carried out using our extension of
the \citet{paisley:2010:stick} stick-breaking construction
in \eq{stick-breaking_two}. We generated 2,000
rounds of feature probabilities; that is, we generated 2,000 random
variables $C_{i}$ and $\sum_{i=1}^{2,000} C_{i}$ feature
probabilities. With these probabilities, we generated $N$ = 1,000
data points, i.e., 1,000 vectors of (2,000) independent Bernoulli random
variables with these probabilities. With these simulated data, we
were able to perform an empirical evaluation of our theoretical results.

\begin{figure}
\center

\includegraphics[width=0.46\textwidth]{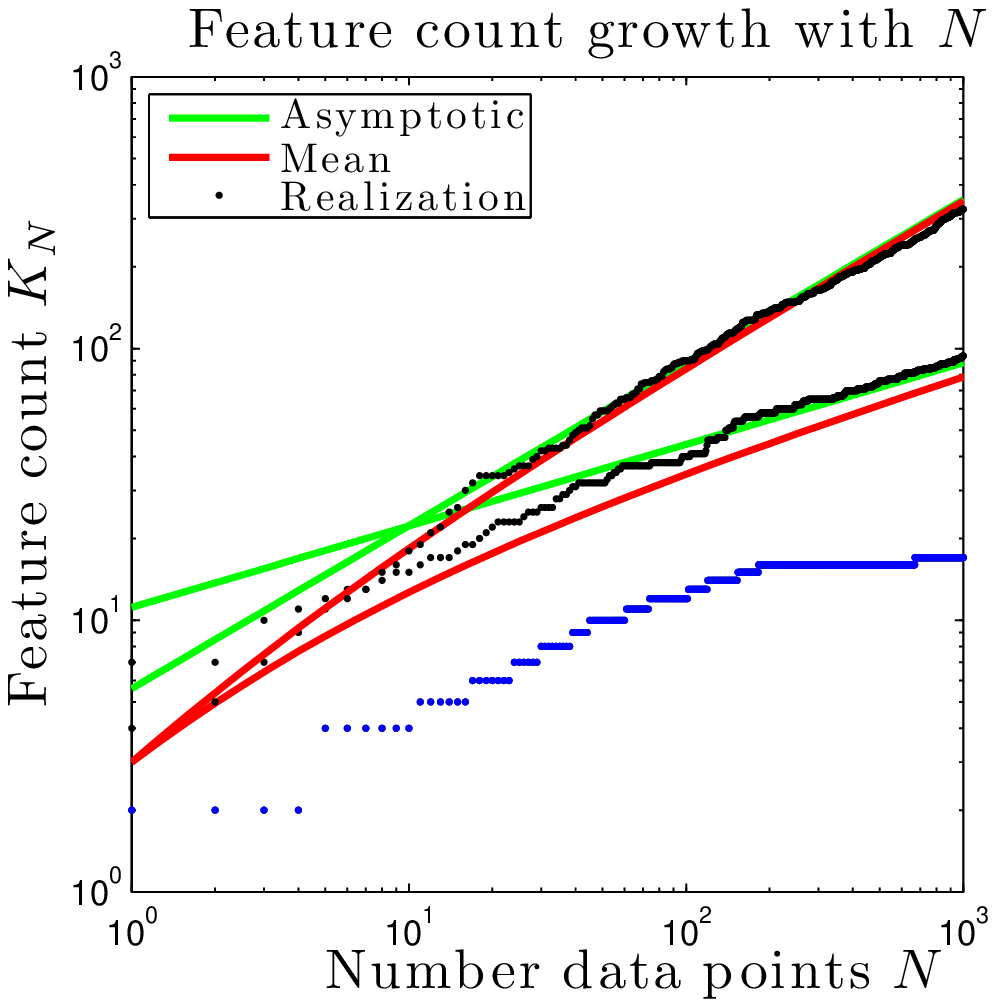}
\includegraphics[width=0.44\textwidth]{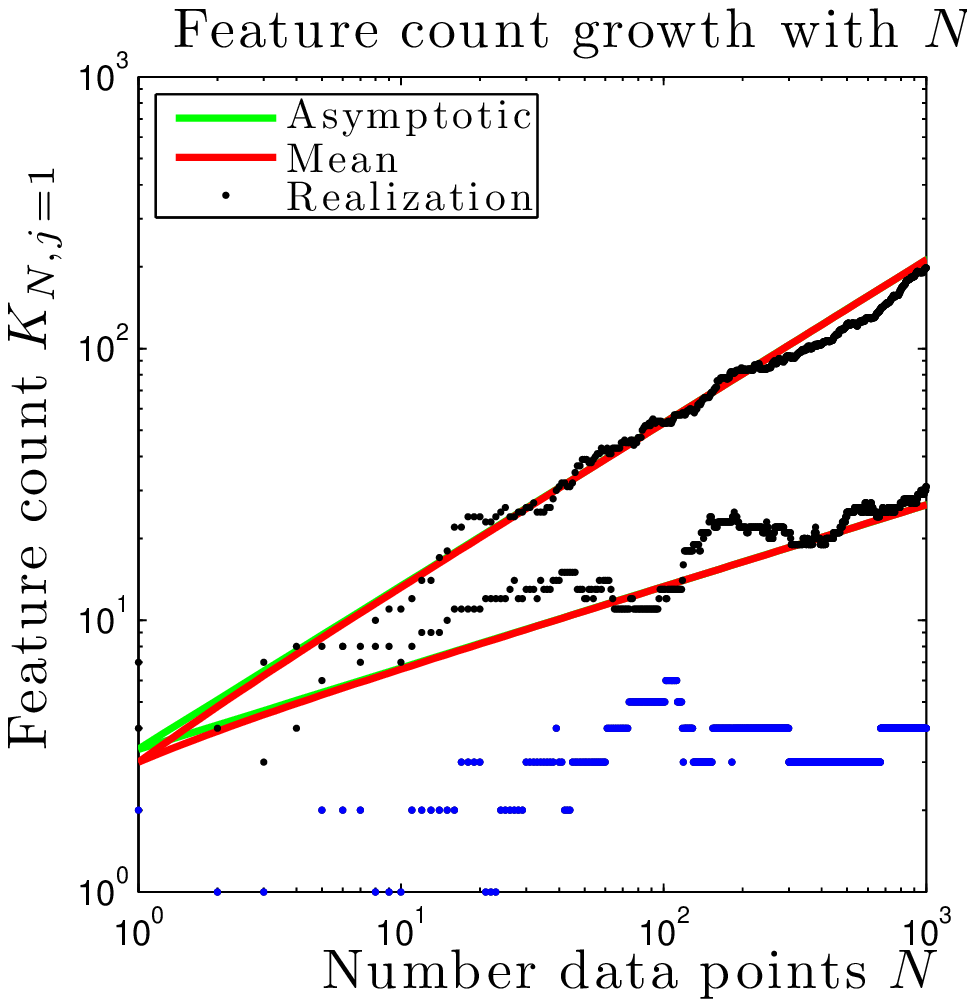}

\caption{\label{fig:sim_K} Growth in the number of represented
features $K_{N}$ ({\em left}) and the number of 
features represented by exactly one data point $K_{N,1}$ ({\em right})
as the total number of data points $N$ grows. The points in the scatterplot are derived
by simulation; blue for $\alpha=0$, center black is $\alpha=0.3$, and upper black for $\alpha=0.6$.
The red lines in the {\em left} plot show the theoretical mean $\Phi_{N}$ (\eq{phi_N});
in the {\em right} plot,
they show the theoretical mean $\Phi_{N,1}$ (\eq{phi_Nj}). The green lines
show the theoretical asymptotic behavior, \eq{power_total_beta} on the {\em left}
(Type I power law) and \eq{power_r_beta} on the {\em right} (Type II power law).
}
\end{figure}

\fig{sim_K} illustrates power laws in the number of represented
features $K_{N}$ on the left (Type I power law) and the number of
features represented by exactly one data point $K_{N,1}$ on the
right (Type II power law).  Both of these quantities are plotted as 
functions of the increasing number of data points $N$.  The blue
points show the simulated values for the classic, two-parameter beta process case
with $\alpha=0$. The center set of black points in each case corresponds to $\alpha=0.3$,
and the upper set of black points in each case corresponds to $\alpha=0.6$.

We also plot curves obtained from our theoretical results in order to
compare them to the simulation.  Recall that in our theoretical development,
we noted that there are two steps to establishing the asymptotic behavior of
$K_{N}$ and $K_{N,j}$ as $N$ increases. First, we compare the random quantities
$K_{N}$ and $K_{N,j}$ to their respective means, $\Phi_{N}$ and $\Phi_{N,j}$.
These means, as computed via numerical quadrature from \eq{phi_N} and
directly from \eq{phi_Nj}, are shown by red curves in the plots. Second,
we compare the means to their own asymptotic behavior. This asymptotic
behavior, which we ultimately proved was shared with the respective $K_{N}$
or $K_{N,j}$ in \eqs{power_total_beta} and \eqss{power_r_beta}, is
shown by green curves in the plots.

We can see in both plots that the $\alpha=0$ behavior is distinctly different 
from the straight-line behavior of the $\alpha>0$ examples. In both cases, we 
can see that any growth in $\alpha$ is slower than can be described by straight-line
growth. In particular, when $\alpha = 0$, the expected number of features is
\begin{equation}
	\label{eq:phi_2param}
	\phi_{N}
		= \mbe[K_{N}] = \mbe\left[ \sum_{n=1}^{N} \pois\left( \gamma \frac{\theta}{n + \theta} \right) \right] \\
		= \sum_{n=1}^{N} \gamma \frac{\theta}{n + \theta} \\
		\sim \gamma \theta \log(N).
\end{equation}
Similarly, when $\alpha = 0$, the expected number of features represented
by exactly one data point, $K_{N,1}$, is (by \eq{phi_Nj})
\begin{align*}
	\Phi_{N,1}
		&= \mbe[K_{N,1}] = \binom{N}{1} \int_{0}^{1} x^{1} (1-x)^{N-1} \cdot \theta x^{-1}(1-x)^{\theta-1} \; dx \\
		&= N \theta \cdot \frac{\Gamma(1) \Gamma(N-1+\theta)}{\Gamma(N+\theta)}
		= \theta \frac{N}{N-1+\theta}
		\sim \theta,
\end{align*}
where the second line follows from using the
normalization constant of the (proper) beta distribution.
Interestingly, while $K_{N,1}$ grows as a power
law when $\alpha > 0$, its expectation is constant
when $\alpha = 0$. While many new features are
instantiated as $N$ increases in the $\alpha=0$ case,
it seems that they are quickly represented by more
data points than just the first one.

\begin{figure}
\center

\includegraphics[width=0.43\textwidth]{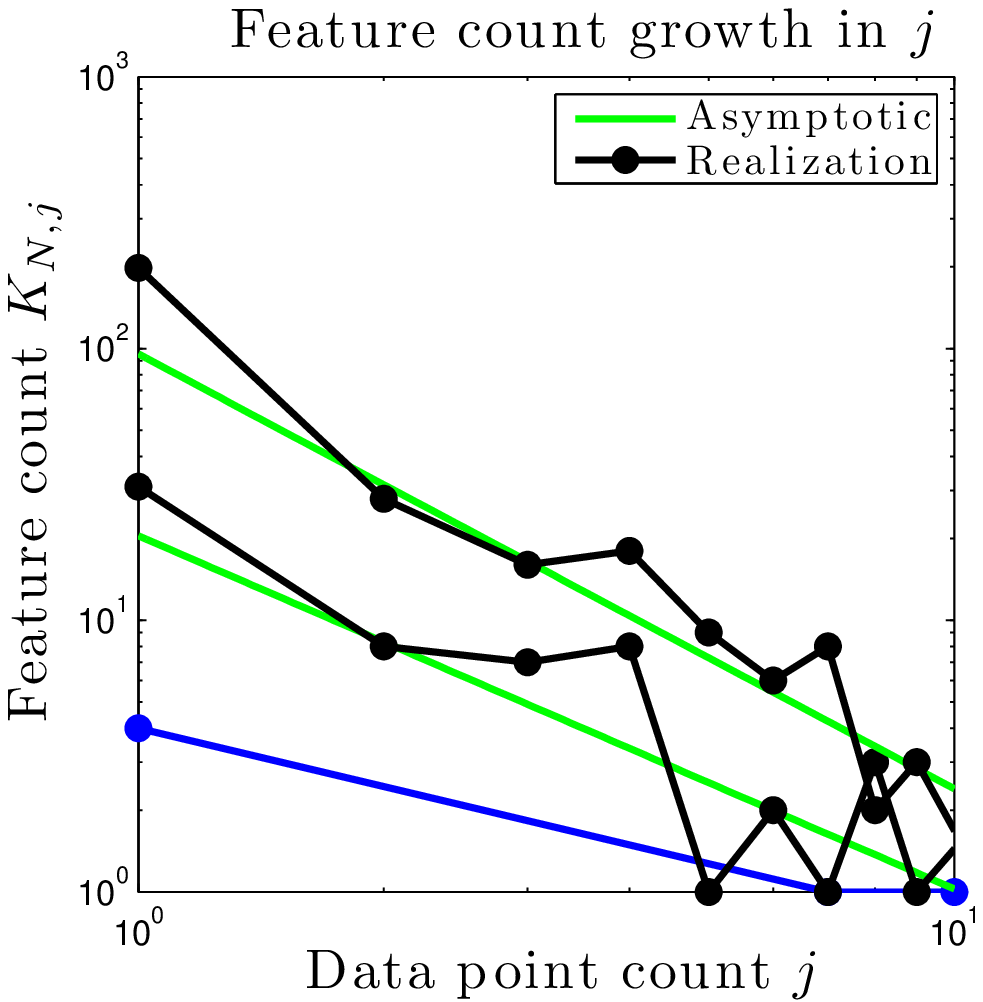}
\includegraphics[width=0.46\textwidth]{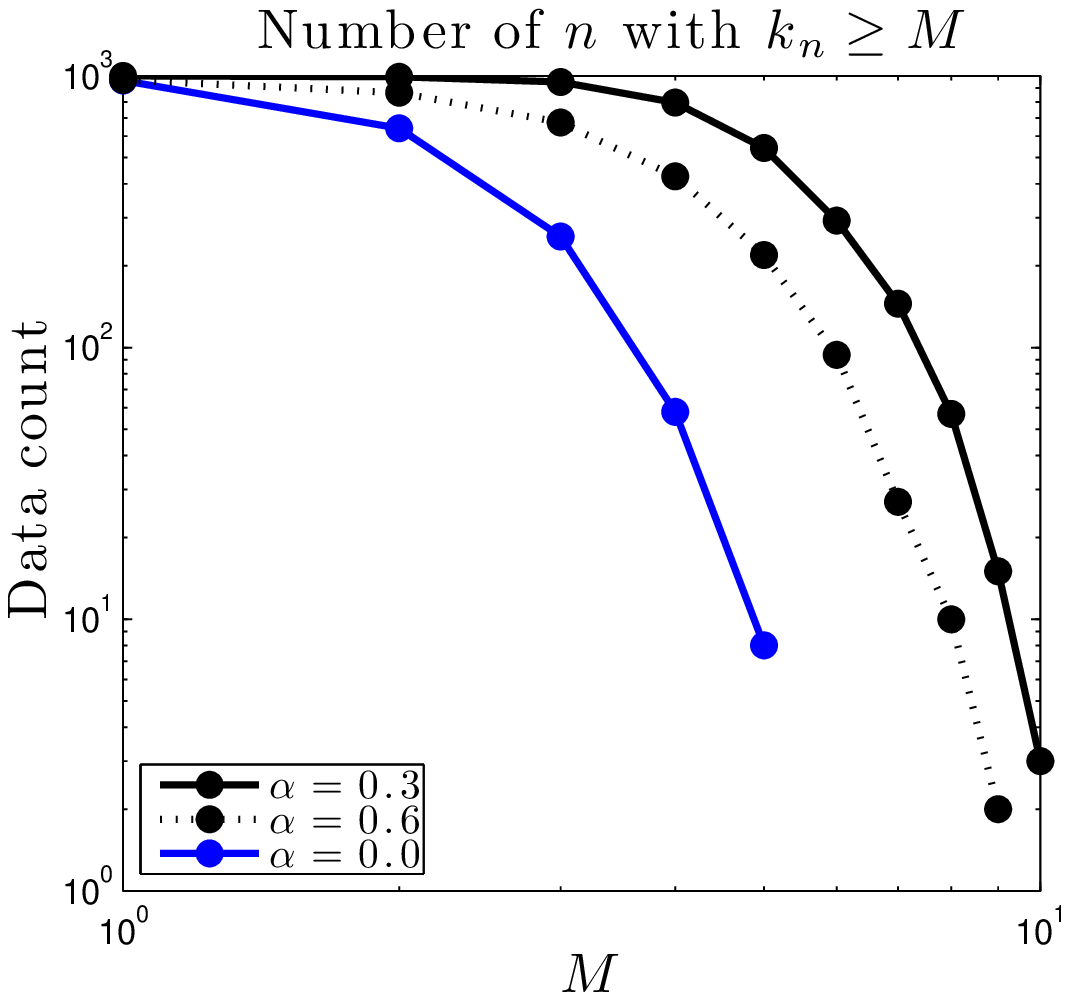}

\caption{\label{fig:sim_j_iii} {\em Left}: Change in the number of
features with exactly $j$ representatives among $N$ data points
for fixed $N$ as a function of $j$. The blue points, with connecting
lines, are for $\alpha=0$; middle black are for $\alpha=0.3$, upper black for
$\alpha=0.6$. The green lines show the theoretical asymptotic behavior
in $j$ (\eqs{zipf_law} and \eqss{asymp_j}) for the two $\alpha > 0$ cases. {\em Right}: 
Change in the number of data points, indexed by $n$, with number
of feature assignments $k_{n}$ greater than some positive, real-valued
$M$ as $M$ increases. Neither the $\alpha=0$ case (blue) nor the
$\alpha>0$ cases (black) exhibit Type III power laws.
}
\end{figure}

Type I and II power laws are somewhat easy to visualize since we 
have one point in our plots for each data point simulated. 
The behavior of $K_{N,j}$ as a function of $j$ for fixed $N$ and type III
power laws (or lack thereof) are somewhat more difficult to
visualize. In the case of $K_{N,j}$ as a function of $j$, we 
might expect that a large number of data points $N$ is
necessary to see many groups of size $j$ for $j$ much greater than
one. In the Type III case, we have seen that in fact power laws do
not hold for any value of $\alpha$ in the beta process. Rather, the number
of data points exhibiting more than $M$ features decreases more quickly in
$M$ than a power law would predict; 
therefore, we cannot plot many values of $M$ before this number effectively
goes to zero.

Nonetheless, \fig{sim_j_iii} compares our simulated data to
the approximation of \eq{zipf_law} with \eq{asymp_j} ({\em left})
and Type III power laws ({\em right}). On the left, blue points as usual denote
simulated data under $\alpha=0$; middle black points show $\alpha=0.3$, and
upper black points show $\alpha=0.6$. Here, we use connecting lines between 
plotted points to clarify $\alpha$ values. The green lines for the $\alpha > 0$ case
illustrate the approximation of \eq{asymp_j}. Around $j=10$, we see
that the number of feaures exhibited by $j$ data points, $K_{N,j}$, 
degenerates to mainly zero and one values. However, for smaller values
of $j$ we can still distinguish the power law trend.

On the right-hand side of \fig{sim_j_iii}, we display the number of data points
exhibiting more than $M$ features for various values of $M$ across the three
values of $\alpha$. Unlike the previous plots in \fig{sim_K} and \fig{sim_j_iii}, there
is no power-law behavior for the cases $\alpha > 0$, as predicted in
\mysec{exp_proof}. We also note that here the $\alpha=0.3$ curve
does not lie between the $\alpha=0$ and $\alpha=0.6$ curves.
Such an occurrence is not unusual in this case since, as we saw in
\eq{no_type_iii}, the rate of decrease is modulated by the total
mass of the feature probabilities drawn from the beta process, 
which is random and not necessarily smaller when $\alpha$ is smaller.

\begin{figure}
\center

\includegraphics[width=0.45\textwidth]{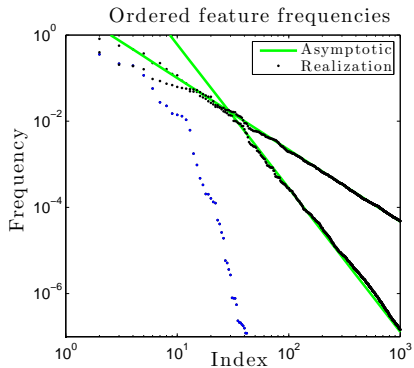}

\caption{\label{fig:sim_freqs} Feature probabilities from the beta
process plotted in decreasing size order. Blue points represent probabilities
from the $\alpha=0$ case; center black points show $\alpha=0.3$, and upper black points
show $\alpha=0.6$. The green lines show theoretical asymptotic behavior of the
ranked probabilities (\eq{theory_freqs}).
}
\end{figure}

Finally, since our experiment involves generating the underlying
feature probabilities from the beta process as well as the actual
feature assignments from repeated draws from the Bernoulli process, 
we may examine the feature probabilities themselves; see
\fig{sim_freqs}. As usual, the blue points represent the
classic, two-parameter ($\alpha=0$) beta process. Black points represent
$\alpha=0.3$ (center) and $\alpha=0.6$ (upper). Perhaps due to the fact 
that there is only the beta process noise to contend with in this aspect 
of the simulation (and not the combined randomness due to the beta process 
and Bernoulli process), we see the most striking demonstration of both 
power law behavior in the $\alpha>0$ cases and faster decay in the 
$\alpha=0$ case in this figure. The two $\alpha>0$ cases clearly adhere
to a power law that may be predicted from our results above and the
\citet{gnedin:2007:notes} results with $C$ as in \eq{C_val}:
\begin{equation}
	\label{eq:theory_freqs}
	\#\{i: q_{i} \ge x\} \stackrel{a.s.}{\sim} C x^{-\alpha} \quad x \downarrow 0.
\end{equation}
Note that ranking the probabilities merely inverts the plot that would be created
with $x$ on the horizontal axis and $\{i: q_{i} \ge x\}$ on the vertical axis.
The simulation demonstrates little noise about these power laws
beyond the 100th ranked probability.
The decay for $\alpha=0$ is markedly faster than the other cases.

\section{Experimental results} \label{sec:experiments}

We have seen that the Poisson process formulation allows for an easy 
extension of the beta process to a three-parameter model.  In this section
we study this model empirically in the setting of the modeling of
handwritten digits.  \citet{paisley:2010:stick} present results
for this problem using a two-parameter beta process coupled with
a discrete factor analysis model; we repeat those experiments with
the three-parameter beta process.  The data consists of 3,000 examples 
of handwritten digits, in particular 1,000 handwriting samples of 
each of the digits 3, 5, and 8 from the MNIST Handwritten Digits 
database~\citep{lecun:1998:mnist,roweis:2007:mnist}. Each handwritten 
digit is represented by a matrix of 28$\times$28 pixels; we project
these matrices into 50 dimensions using principal components 
analysis. Thus, our data takes the form $X \in \mathbb{R}^{50 \times 3000}$, 
and we may apply the beta process factor model from \eq{factor} with 
$P = 50$ and $N$ = 3,000 to discover latent structure in this data.

The generative model for $X$ that we use is
as follows~\citep[see][]{paisley:2010:stick}:
\begin{align}
        \nonumber
	X &= (W \circ Z) \Phi + E \\
        \nonumber
        Z &\sim \bpbep(N,\gamma,\theta,\alpha) \\
	\nonumber
	\Phi_{k,p} &\stackrel{iid}{\sim} N(0, \rho_{p}) \\
	\nonumber
        W_{n,k} &\stackrel{iid}{\sim} N(0, \zeta) \\
        \label{eq:gen_model}
        E_{n,p} &\stackrel{iid}{\sim} N(0, \eta),
\end{align}
with hyperparameters $\theta,\alpha,\gamma,B_{0},\{\rho_{p}\}_{p=1}^{P}, \zeta, \eta$.
Recall from \eq{factor} that $X \in \mathbb{R}^{N \times P}$ is the data,
$\Phi \in \mathbb{R}^{K \times P}$ is a matrix of factors,
and $E \in \mathbb{R}^{N \times P}$ is an error matrix. Here, we
introduce the weight matrix $W \in \mathbb{R}^{N \times K}$, which
modulates the binary factor loadings $Z \in \mathbb{R}^{N \times K}$.
In \eq{gen_model}, $\circ$ denotes elementwise multiplication, and the indices
have ranges $n \in \{1,\ldots,N\}, k \in \{1,\ldots,K\}, p \in \{1,\ldots,P\}$.
Since we draw $Z$ from a beta-Bernoulli process, the dimension $K$ is
theoretically infinite in the generative model notation of \eq{gen_model}.
However, we have seen that the number of columns of $Z$ with
nonzero entries is finite a.s. We use $K$ to denote this number.

We initialized both the two-parameter and the three-parameter models
with the same number of latent features, $K=200$, and the same values 
for all shared parameters (i.e., every variable except the new discount 
parameter $\alpha$).  We ran the experiment for 2,000 MCMC iterations, noting 
that the MCMC runs in both models seem to have reached equilibrium by 
500 iterations (see \figs{K} and \figss{hyper}).

\begin{figure}
\begin{multicols}{2}
\center

\includegraphics[width=0.45\textwidth]{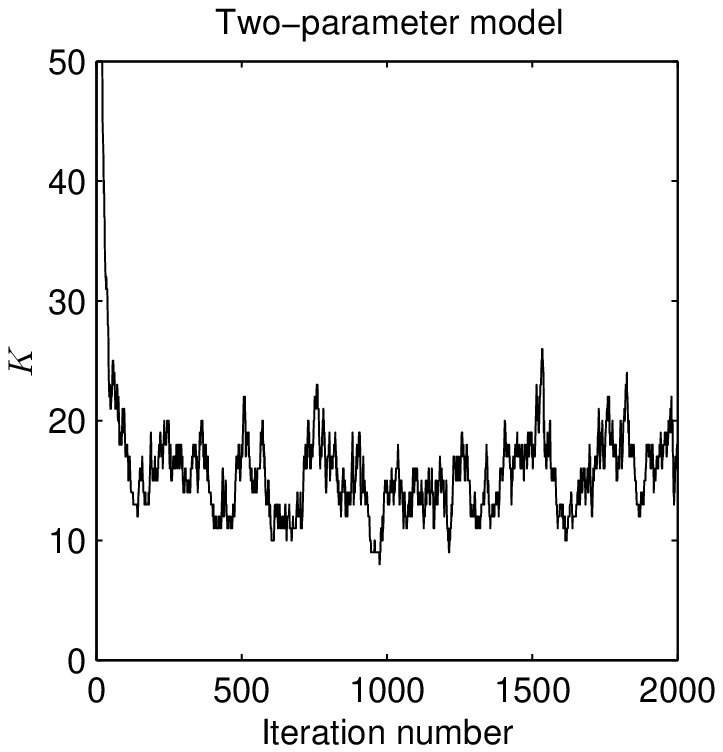}

\columnbreak

\includegraphics[width=0.45\textwidth]{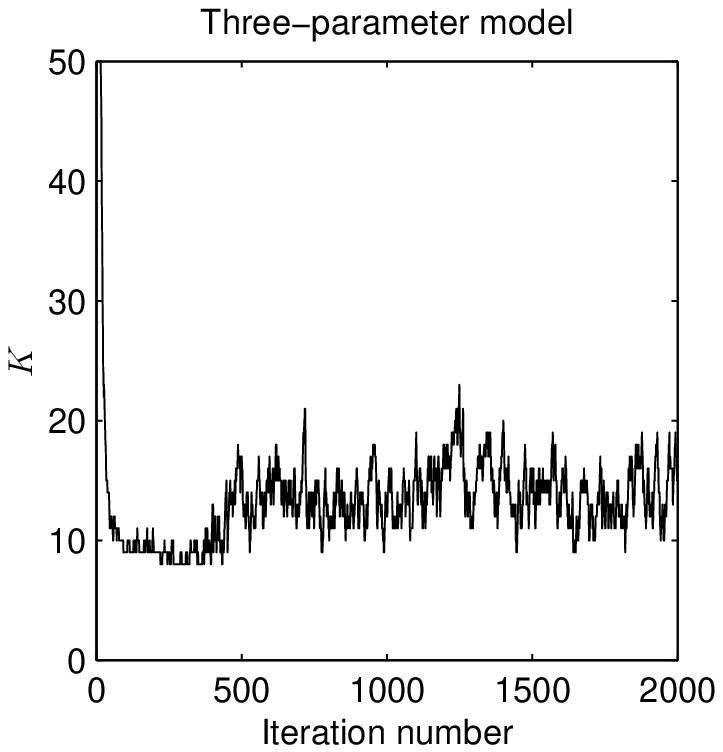}

\end{multicols}
\caption{\label{fig:K} The number of latent features $K$ as a function 
of the MCMC iteration. Results for the original, two-parameter model
are represented on the {\em left}, and results for the new, three-parameter 
model are illustrated on the {\em right}.}
\end{figure}

\begin{figure}
\begin{multicols}{3}
\center

\includegraphics[width=0.3\textwidth]{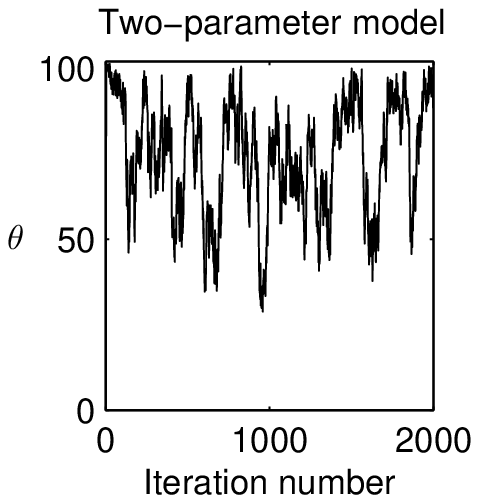}

\columnbreak

\includegraphics[width=0.3\textwidth]{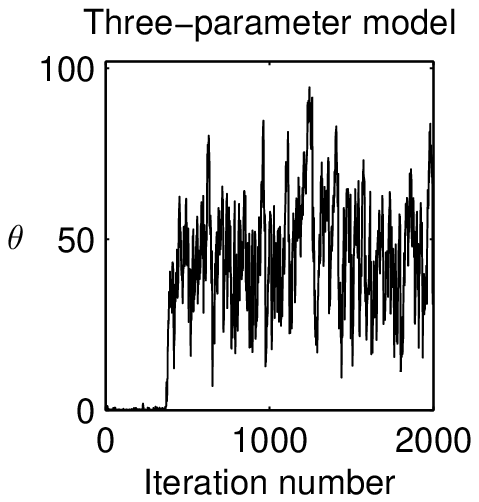}

\columnbreak

\includegraphics[width=0.3\textwidth]{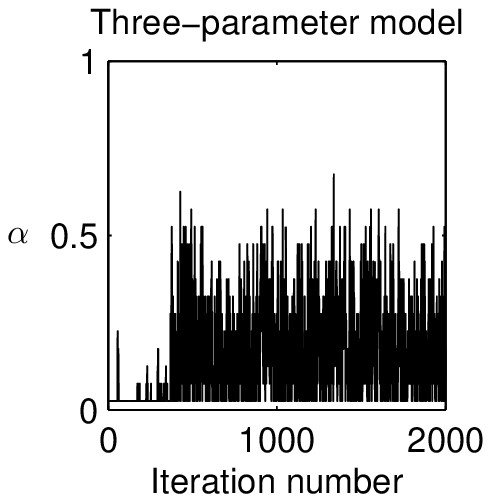}

\end{multicols}
\caption{\label{fig:hyper} The random values drawn for the hyperparameters 
as a function of the MCMC iteration. Draws for the concentration parameter 
$\theta$ under the two-parameter model are shown on the {\em left}, 
and draws for $\theta$ under the three-parameter model are shown in 
the {\em middle}.  On the \emph{right} are draws of the new discount 
parameter $\alpha$ under the three-parameter model.}
\end{figure}

\begin{figure}
\begin{multicols}{2}
\center

\includegraphics[width=0.45\textwidth]{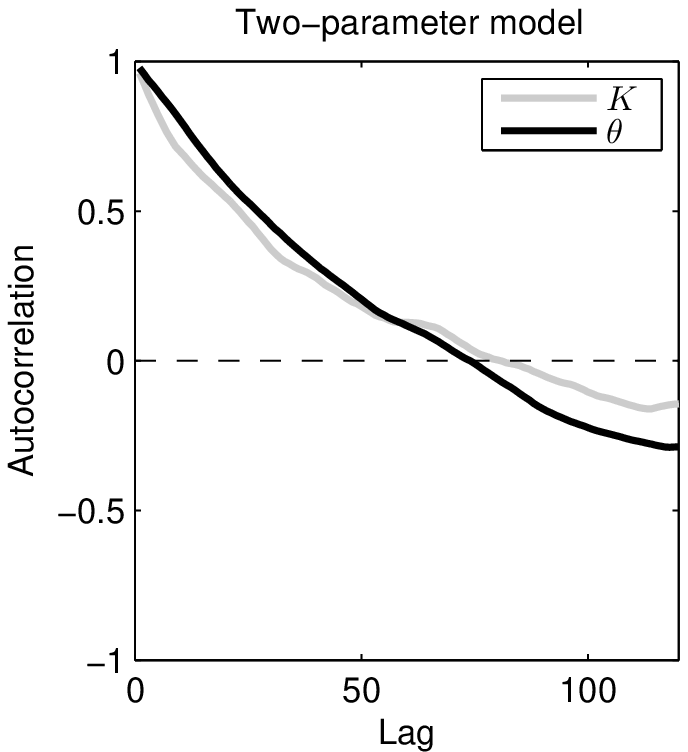}

\columnbreak

\includegraphics[width=0.45\textwidth]{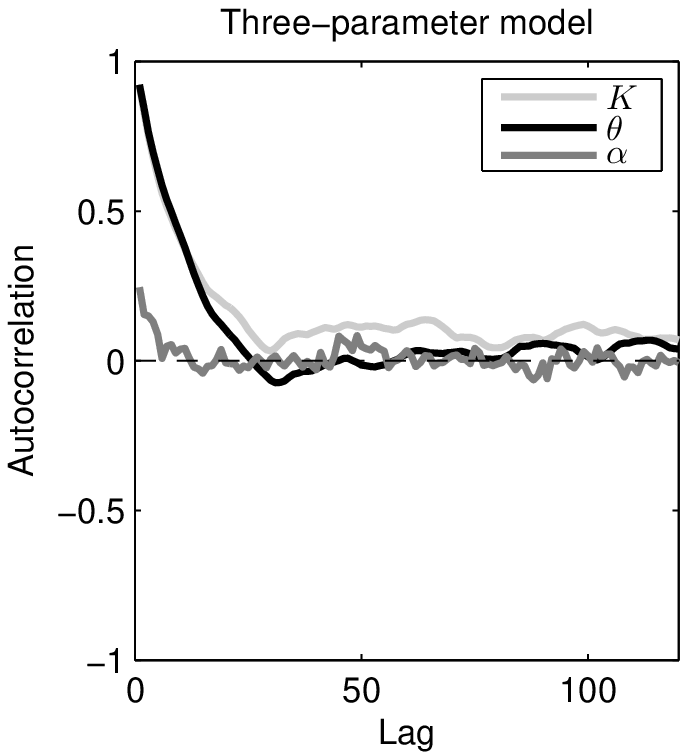}

\end{multicols}
\caption{\label{fig:acorr} Autocorrelation of the number of factors 
$K$, concentration parameter $\theta$, and discount parameter $\alpha$ 
for the MCMC samples after burn-in (where burn-in is taken to end at 
500 iterations) under the two-parameter model ({\em left}) and 
three-parameter model ({\em right}).} \end{figure}

\figs{K} and \figss{hyper} show the sampled values of 
various parameters as a function of MCMC iteration. In particular, 
we see how the number of features $K$ (\fig{K}), the 
concentration parameter $\theta$, and the discount parameter $\alpha$ 
(\fig{hyper}) change over time.  All three graphs 
illustrate that the three-parameter model takes a longer time 
to reach equilibrium than the two-parameter model (approximately 
500 iterations vs.\ approximatively 100 iterations).  However, 
once at equilibrium, the sampling time series associated with 
the three-parameter iterations exhibit lower autocorrelation 
than the samples associated with the two-parameter iterations 
(\fig{acorr}). In the implementation of both the original 
two-parameter model and the three-parameter model, the range for 
$\theta$ is considered to be bounded above by approximately 100 
for computational reasons (in accordance with the original 
methodology of~\citet{paisley:2010:stick}).  As shown in 
\fig{hyper}, this bound affects sampling in the 
two-parameter experiment whereas, after burn-in, the effect 
is not noticeable in the three-parameter experiment. 
While the discount parameter $\alpha$ also comes close to the 
lower boundary of its discretization (\fig{hyper})---which 
cannot be exactly zero due to computational concerns---the samples 
nonetheless seem to explore the space well.

We can see from \fig{acorr} that the estimated value of the concentraton
parameter $\theta$ is much lower when the discount parameter $\alpha$
is also estimated. This behavior may be seen to result from the fact that
the power law growth of the expected number of represented features $\Phi_{N}$
in the $\alpha > 0$ case yields a generally higher
expected number of features than in the $\alpha = 0$ case for a
fixed concentration parameter $\theta$. Further, we see from \eq{phi_2param}
that the expected number of features when $\alpha = 0$ is linear in $\theta$.
Therefore, if we instead fix the number of features,
the $\alpha = 0$ model can compensate by increasing $\theta$ over
the $\alpha > 0$ model. Indeed, we see in \fig{K} that the number of features
discovered by both models is roughly equal; in order to achieve this number
of features, the $\alpha = 0$ model seems to be compensating by overestimating 
the concentration parameter $\theta$.

\begin{figure}
\begin{center} 

Two-parameter model

\includegraphics[width= 0.1\textwidth]{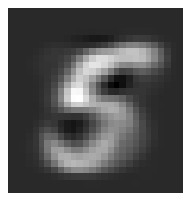}
\includegraphics[width= 0.1\textwidth]{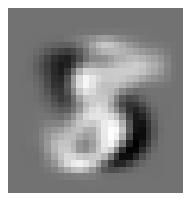}
\includegraphics[width= 0.1\textwidth]{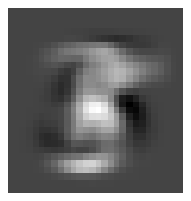}
\includegraphics[width= 0.1\textwidth]{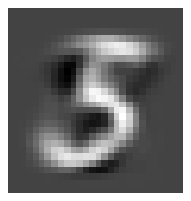}
\includegraphics[width= 0.1\textwidth]{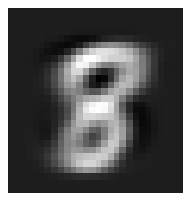}
\includegraphics[width= 0.1\textwidth]{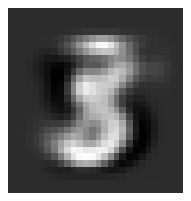}
\includegraphics[width= 0.1\textwidth]{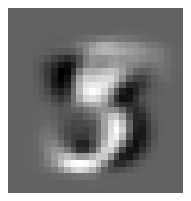}
\includegraphics[width= 0.1\textwidth]{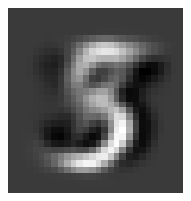}
\includegraphics[width= 0.1\textwidth]{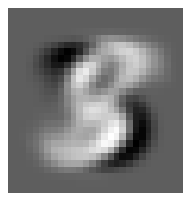}

Three-parameter model

\includegraphics[width= 0.1\textwidth]{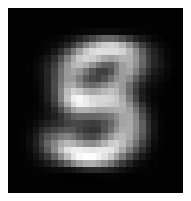}
\includegraphics[width= 0.1\textwidth]{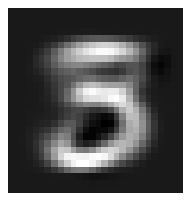}
\includegraphics[width= 0.1\textwidth]{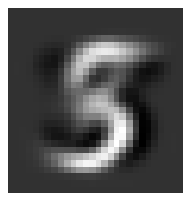}
\includegraphics[width= 0.1\textwidth]{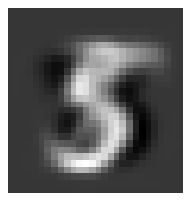}
\includegraphics[width= 0.1\textwidth]{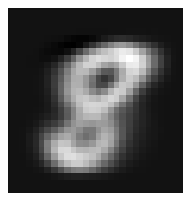}
\includegraphics[width= 0.1\textwidth]{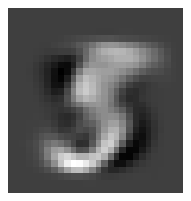}
\includegraphics[width= 0.1\textwidth]{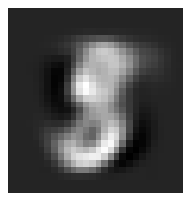}
\includegraphics[width= 0.1\textwidth]{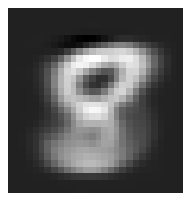}
\includegraphics[width= 0.1\textwidth]{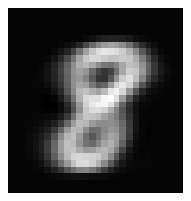}
\end{center}

\caption{\label{fig:factor} {\em Upper}: The top nine features by 
sampled representation across the data set on the final MCMC 
iteration for the original, two-parameter model. 
{\em Lower}: The top nine features determined in the same 
way for the new, three-parameter model.}
\end{figure}

To get a sense of the actual output of the model, we can look 
at some of the learned features. In particular, we collected 
the set of features from the last MCMC iteration in each model. 
The $k$th feature is expressed or not for the $n$th data point 
according to whether $Z_{nk}$ is one or zero. Therefore, we 
can find the most-expressed features across the data set using 
the set of features on this iteration as well as the sampled 
$Z$ matrix on this iteration.  We plot the nine most-expressed
features under each model in \fig{factor}. In both 
models, we can see how the features have captured distinguishing 
features of the 3, 5, and 8 digits.

Finally, we note that the three-parameter version of the algorithm 
is competitive with the two-parameter version in running time once 
equilibrium is reached. After the burn-in regime of 500 iterations, 
the average running time per iteration under the three-parameter 
model is 14.5 seconds, compared with 11.7 seconds average running 
time per iteration under the two-parameter model.

\section{Conclusions} \label{sec:conclusion}

We have shown that the stick-breaking representation of the
beta process due to~\citet{paisley:2010:stick} can be obtained
directly from the representation of the beta process as a 
completely random measure.  With this result in hand the
set of connections between the beta process, stick-breaking,
and the Indian buffet process are essentially as complete as
those linking the Dirichlet process, stick-breaking, and the
Chinese restaurant process.

We have also shown that this approach motivates a three-parameter 
generalization of the stick-breaking representation of~\citet{paisley:2010:stick}, 
which is the analog of the Pitman-Yor generalization of the stick-breaking 
representation for the Dirichlet process.  We have shown that Type I
and Type II power laws follow from this three-parameter model.
We have also shown that Type III power laws cannot be obtained
within this framework.  It is an open problem to discover useful
classes of stochastic processes that provide such power laws.

\section{Acknowledgments} \label{sec:acknowledge}

We wish to thank Alexander Gnedin for useful discussions
and Lancelot James for helpful suggestions.
We also thank John Paisley for useful discussions and for kindly 
providing access to his code, which we used in our experimental work.  
Tamara Broderick was funded by a National Science Foundation 
Graduate Research Fellowship.  Michael Jordan was supported in
part by IARPA-BAA-09-10, ``Knowledge Discovery and Dissemination.''  
Jim Pitman was supported in part by 
the National Science Foundation Award 0806118 ``Combinatorial 
Stochastic Processes.'' 

\appendix

\section{A Markov chain Monte Carlo algorithm} \label{app:inference}

Posterior inference under the three-parameter model can be performed with a 
Markov chain Monte Carlo (MCMC) algorithm. Many conditionals have simple forms 
that allow Gibbs sampling although others require further approximation. Most of 
our sampling steps are as in~\citet{paisley:2010:stick} with the
notable exceptions of a new sampling step for the discount parameter $\alpha$
and integration of the discount parameter $\alpha$
into the existing framework. We describe the full algorithm here.

\subsection{Notation and auxiliary variables}

Call the index $i$ in \eq{stick-breaking_two} the {\em round}. Then introduce the round-indicator variables $r_{k}$ such that $r_{k} = i$ exactly when the $k$th atom, where $k$ indexes the sequence $(\psi_{1,1},\ldots,\psi_{1,C_{1}},\psi_{2,1},\ldots,\psi_{2,C_{2}},\ldots)$, occurs in round $i$. We may write
$$
r_{k} := 1 + \sum_{i=1}^{\infty} \mathbbm{1}\left\{ \sum_{j=1}^{i} C_{j} < k \right\}.
$$
To recover the round lengths $C$ from $r = (r_1, r_2, \ldots)$, note that
\begin{equation}
	\label{eq:round_lengths_from_indicators} 
	C_{i} = \sum_{k=1}^{\infty} \mathbbm{1}(r_{k} = i).
\end{equation}

With the definition of the round indicators $r$ in hand, we can rewrite the beta process $B$ as
$$
	B = \sum_{k=1}^{\infty} V_{k,r_{k}} \prod_{j=1}^{r_{k}} (1-V_{k,j}) \delta_{\psi_{k}},
$$
where $V_{k,j} \stackrel{iid}{\sim} \tb(1-\alpha,\theta+i\alpha)$ and $\psi_{k} \stackrel{iid}{\sim} \gamma^{-1} B_{0}$ as usual although the indexing is not the same as in \eq{stick-breaking_two}.
It follows that the expression of the $k$th feature for the $n$th data point is
given by
$$
	Z_{n,k} \sim \bern\left(\pi_{k} \right), \quad
	\pi_{k} := V_{k,r_{k}} \prod_{j=1}^{r_{k}-1}(1-V_{k,j}).
$$

We also introduce notation for the number of data points in which the $k$th feature is, respectively, expressed and not expressed:
$$
	m_{1,k} := \sum_{n=1}^{N} \mathbbm{1}(Z_{n,k}=1), \quad
	m_{0,k} := \sum_{n=1}^{N} \mathbbm{1}(Z_{n,k}=0)
$$
Finally, let $K$ be the number of represented features; i.e., $K := \#\{k: m_{1,k} > 0\}$. Without loss of generality, we assume the represented features are the first $K$ features in the index $k$. The new quantities $\{r_{k}\}$, $\{m_{1,k}\}$, $\{m_{0,k}\}$, and $K$ will be used in describing the sampler steps below.

\subsection{Latent indicators}

First, we describe the sampling of the round indicators $\{r_{k}\}$ and the latent feature indicators $\{Z_{n,k}\}$. In these and other steps in the MCMC algorithm, we integrate out the stick-breaking proportions $\{V_{i}\}$.

\subsubsection{Round indicator variables} \label{app:round_indicator}

We wish to sample the round indicator $r_{k}$ for each feature $k$ with $1 \le k \le K$. We can write the conditional for $r_{k}$ as
\begin{eqnarray}
	\nonumber
	\lefteqn{ p(r_{k} = i | \{r_{l}\}_{l=1}^{k-1}, \{Z_{n,k}\}_{n=1}^{N}, \theta, \alpha, \gamma) } \\
		\label{eq:sample_round_indicator}
		&\propto& p(\{Z_{n,k}\}_{n=1}^{N} | r_{k} = i, \theta, \alpha) p(r_{k} = i | \{r_{l}\}_{l=1}^{k-1}).
\end{eqnarray}
It remains to calculate the two factors in the product.

For the first factor in \eq{sample_round_indicator}, we write out the integration over stick-breaking proportions and approximate with a Monte Carlo integral:
\begin{eqnarray}
	\nonumber
	p(\{Z_{n,k}\}_{n=1}^{N} | r_{k} = i, \theta, \alpha)
		&=& \int_{[0,1]^{i}} \pi_{k}^{m_{1,k}} (1-\pi_{k})^{m_{0,k}} \; dV \\
		\label{eq:mc_int_sticks}
		&\approx& \frac{1}{S} \sum_{s=1}^{S} (\pi_{k}^{(s)})^{m_{1,k}} (1-\pi_{k}^{(s)})^{m_{0,k}}.
\end{eqnarray}
Here, $\pi^{(s)}_{k} := V^{(s)}_{k,r_{k}} \prod_{j=1}^{r_{k}-1}(1-V^{(s)}_{k,j})$, and $V_{k,j}^{(s)} \stackrel{indep}{\sim} \tb(1-\alpha,\theta + j\alpha)$. Also, $S$ is the number of samples in the sum approximation. Note that the computational trick employed in~\citet{paisley:2010:stick} for sampling the $\{V_{i}\}$ relies on the first parameter of the beta distribution being equal to one; therefore, the sampling described above, without further tricks, is exactly the sampling that must be used in this more general parameterization. 

For the second factor in \eq{sample_round_indicator}, there is no dependence on the $\alpha$ parameter, so the draws are the same as in~\citet{paisley:2010:stick}. For $R_{k} := \sum_{j=1}^{k} \mathbbm{1}(r_{j} = r_{k})$, we have
\begin{eqnarray*}
	\lefteqn{ p(r_{k} = r | \gamma, \{r_{l}\}_{l=1}^{k-1}) } \\
		&=& \left\{ \begin{array}{ll}
			0 & r < r_{k-1} \\
			\frac{1 - \sum_{i=1}^{R_{k-1}} \pois(i | \gamma)}{
				1 - \sum_{i=1}^{R_{k-1}-1} \pois(i | \gamma)}
				& r = r_{k-1} \\
			\left( 1 - \frac{1 - \sum_{i=1}^{R_{k-1}} \pois(i | \gamma)}{
				1 - \sum_{i=1}^{R_{k-1}-1} \pois(i | \gamma)} \right)
				\left( 1 - \pois(0 | \gamma) \right) \pois( 0 | \gamma )^{h-1}
				& r = r_{k-1} + h
			\end{array} \right.
\end{eqnarray*}
for each $h \ge 1$. Note that these draws make the approximation
that the first $K$ features correspond
to the first $K$ tuples $(i,j)$ in the double sum of \eq{stick-breaking_two};
these orderings do not in general agree.

To complete the calculation of the posterior for $r_{k}$, we need to sum over all values of $i$ to normalize $p(r_{k} = i | \{r_{l}\}_{l=1}^{k-1}, \{Z_{n,k}\}_{n=1}^{N}, \theta, \alpha, \gamma)$. Since this is not computationally feasible, an alternative method is to calculate \eq{sample_round_indicator} for increasing values of $i$ until the result falls below a pre-determined threshold.

\subsubsection{Factor indicators}

In finding the posterior for the $k$th feature indicator in the $n$th latent 
factor, $Z_{n,k}$, we can integrate out both $\{V_{i}\}$ and the weight variables 
$\{W_{n,k}\}$. The conditional for $Z_{n,k}$ is
\begin{eqnarray}
	\nonumber
	\lefteqn{ p(Z_{n,k} | X_{n,\cdot}, \Phi, Z_{n,-k}, r, \theta, \alpha, \eta, \zeta) } \\
		\label{eq:factor_indicator}
		&=& p(X_{n,\cdot} | Z_{n,\cdot}, \Phi, \eta, \zeta) p(Z_{n,k} | r, \theta, \alpha, Z_{n,-k}).
\end{eqnarray}

First, we consider the likelihood. For this factor, we integrate out 
$W$ explicitly:
\begin{eqnarray*}
	\lefteqn{ p(X_{n,\cdot} | Z_{n,\cdot}, \Phi, \eta, \zeta) } \\
		&=& \int_{W} p(X_{n,\cdot} | Z_{n,\cdot}, \Phi, W, \eta) p(W | \zeta) \\
		&=& \int_{W_{n,I}} N(X_{n,\cdot} | W_{n,I} \Phi_{I,\cdot}, \eta I_{P} ) N(W_{n,I} | 0_{|I|}, \zeta I_{|I|}) dW_{n,I} \\
		&& \textrm{where $I = \{i: Z_{n,i} = 1\}$} \\
		&=& N\left(X_{n,\cdot} | 0_{P}, \left[  \eta^{-1} I_{P} - \eta^{-2} \Phi_{I,\cdot} \left( \eta^{-1} \Phi_{I,\cdot}^{\top} \Phi_{I,\cdot} + \zeta^{-1} I_{|I|} \right)^{-1} \Phi_{I,\cdot}^{\top} \right]^{-1} \right) \\
		&=& N\left(X_{n,\cdot} | 0_{P}, \eta I_{P} + \zeta \Phi_{I,\cdot} \Phi_{I,\cdot}^{\top} \right),
\end{eqnarray*}
where the final step follows from the Sherman-Morrison-Woodbury lemma.

For the second factor in \eq{factor_indicator}, we can write
\begin{eqnarray*}
	p(Z_{n,k} | r, \theta, \alpha, Z_{n,-k})
		&=& \frac{p(Z_{n} | r, \theta, \alpha)}{p(Z_{n,-k} | r, \theta, \alpha)},
\end{eqnarray*}
and the numerator and denominator can both be estimated as integrals over 
$V$ using the same Monte Carlo integration trick as in \eq{mc_int_sticks}.

\subsection{Hyperparameters}

Next, we describe sampling for the three parameters of the beta process. 
The mass and concentration parameters are shared by the two-parameter process; 
the discount parameter is unique to the three-parameter beta process.

\subsubsection{Mass parameter}

With the round indicators $\{r_{k}\}$ in hand as from \app{round_indicator} above, we can recover the round lengths $\{C_{i}\}$ with \eq{round_lengths_from_indicators}. Assuming an improper gamma prior on $\gamma$---with both shape and inverse scale parameters equal to zero---and recalling the iid Poisson generation of the $\{C_{i}\}$, the posterior for $\gamma$ is
$$
	p(\gamma | r, Z, \theta, \alpha)
		= \ga(\gamma | \sum_{i=1}^{r_{K}} C_{i}, r_{K}).
$$
Note that it is necessary to sample $\gamma$ since it occurs in, e.g., the 
conditional for the round indicator variables (\app{round_indicator}).

\subsubsection{Concentration parameter}

The conditional for $\theta$ is
\begin{equation*}
	p(\theta | Z, r, \alpha)
		\propto p(\theta) \prod_{k=1}^{K} p(Z | r, \theta, \alpha).
\end{equation*}

Again, we calculate the likelihood factors $p(Z | r, \theta, \alpha)$ with a Monte Carlo approximation as in \eq{mc_int_sticks}. In order to find the conditional over $\theta$ from the likelihood and prior, we further approximate the space of $\theta > 0$ by a discretization around the previous value of $\theta$ in the Monte Carlo sampler: $\{\theta_{prev} + t \Delta \theta\}_{t=S}^{t=T}$, where $S$ and $T$ are chosen so that all potential new $\theta$ values are nonnegative and so that the tails of the distribution fall below a pre-determined threshold. To complete the description, we choose the improper prior $p(\theta) \propto 1$.

\subsubsection{Discount parameter}

We sample the discount parameter $\alpha$ in a similar manner to $\theta$. The conditional for $\alpha$ is
\begin{equation*}
	p(\alpha | Z, r, \theta)
		\propto p(\alpha) \prod_{k=1}^{K} p(Z | r, \theta, \alpha).
\end{equation*}
As usual, we calculate the likelihood factors $p(Z | r, \theta, \alpha)$ with a Monte Carlo approximation as in \eq{mc_int_sticks}. While we discretize the sampling of $\alpha$ as we did for $\theta$, note that sampling $\alpha$ is more straightforward since $\alpha$ must lie in $[0,1]$. Therefore, the choice of $\Delta \alpha$ completely characterizes the discretization of the interval. In particular, to avoid endpoint behavior, we consider new values of $\alpha$ among $\{\Delta \alpha / 2 + t \Delta \alpha\}_{t=0}^{(\Delta \alpha)^{-1} - 1}$. Moreover, the choice of $p(\alpha) \propto 1$ is, in this case, a proper prior for $\alpha$.

\subsection{Factor analysis components}

In order to use the beta process as a prior in the factor analysis model 
described in \eq{factor}, we must also describe samplers for the feature 
matrix $\Phi$ and weight matrix $W$.

\subsubsection{Feature matrix}

The conditional for the feature matrix $\Phi$ is
\begin{eqnarray*}
	p(\Phi_{\cdot,p} | X, W, Z, \eta, \rho_{p})
		&\propto& p(X_{\cdot,p} | \Phi_{\cdot,p}, W, Z, \eta I_{N}) p(\Phi_{\cdot,p} | \rho_{p}) \\
		&=& N(X_{\cdot,p} | (W \circ Z) \Phi_{\cdot,p}, \eta I_{N}) N(\Phi_{\cdot,p} | 0_{K}, \rho_{p} I_{K}) \\
		&\propto& N\left(\Phi_{\cdot,p} | \mu, \Sigma \right),
\end{eqnarray*}
where, in the final line, the variance is defined as follows:
$$
	\Sigma := \left( \eta^{-1} (W \circ Z)^{\top} (W \circ Z) + \rho_{p}^{-1} I_{K} \right)^{-1},
$$
and similarly for the mean:
$$
	\mu := \Sigma \eta^{-1} (W \circ Z)^{\top} X_{\cdot,p}.
$$

\subsubsection{Weight matrix}

Let $I = \{i: Z_{n,i} = 1\}$. Then the conditional for the weight matrix $W$ is
\begin{eqnarray*}
	p(W_{n,I} | X, Z, \Phi, \eta)
		&\propto& p(X_{n,\cdot} | \Phi_{I,\cdot}, W_{n,I}, \eta) p(W_{n,I} | \zeta) \\
		&=& N(X_{n,\cdot} | W_{n,I} \Phi_{I,\cdot}, \eta I_{p}) N(W_{n,I} | 0_{|I|}, \zeta I_{|I|}) \\
		&\propto& N(W_{n,I} | \tilde{\mu}, \tilde{\Sigma}),
\end{eqnarray*}
where, in the final line, the variance is defined as $\tilde{\Sigma} := \left( \eta^{-1} \Phi_{I,\cdot} \Phi_{I,\cdot}^{\top} + \zeta^{-1} I_{|I|} \right)^{-1}$, and the mean is defined as $\tilde{\mu} := \tilde{\Sigma} \eta^{-1} X_{n,\cdot} \Phi_{I,\cdot}^{\top}$.

\bibliographystyle{plainnat}
\bibliography{beta}

\begin{thebibliography}{35}
\bibitem[Freedman(1973)]{freedman:1973:another}
D.~Freedman.
\newblock Another note on the {B}orel--{C}antelli lemma and the strong law,
  with the {P}oisson approximation as a by-product.
\newblock \emph{The Annals of Probability}, 1\penalty0 (6):\penalty0 910--925,
  1973.
\end{thebibliography}

\end{document}